%% file: main.tex
\documentclass[12pt,reqno]{amsart}

\usepackage[left=1in, right=1in, top=1.1in,bottom=1.1in]{geometry}
\usepackage[T1]{fontenc}
\usepackage{amsthm,amssymb,amsmath,amsfonts,mathrsfs,amscd}
\usepackage{enumitem}
\usepackage{stmaryrd}
\usepackage{mathrsfs}
\usepackage{pdfsync}
\usepackage[font={scriptsize}]{caption}

\usepackage{graphicx}
\usepackage[font=small]{caption}
\usepackage[font=small]{subcaption}

\usepackage{hyperref}
\hypersetup{
  colorlinks   = true,
  citecolor    = green,
  linkcolor    = red
}

\newtheorem{theorem}{Theorem}[section]

\newtheorem{corollary}[theorem]{Corollary}

\newtheorem{proposition}[theorem]{Proposition}

\theoremstyle{remark}
\newtheorem{remark}[theorem]{Remark}

\theoremstyle{remark}
\newtheorem{example}[theorem]{Example}

\newtheorem{assumption}[theorem]{Assumption}

\newcommand{\inn}[1]{\left\langle #1\right\rangle}
\newcommand{\1}{\pmb{1}}
\newcommand{\tL}{\tilde{\ell}}

\title{A mathematical framework for modelling CLMM dynamics in continuous time}
\author[S.-N. Tung]{Shen-Ning Tung}
\author[T.-H. Wang]{Tai-Ho Wang}
\date{\today}

\address{Shen-Ning Tung \newline
Department of Mathematics, \newline
National Tsing Hua University \newline
No. 101, Sec. 2, Guangfu Rd., Hsinchu, Taiwan
}
\email{tung@math.nthu.edu.tw}

\address{Tai-Ho Wang \newline
Department of Mathematics \newline
Baruch College, The City University of New York \newline
1 Bernard Baruch Way, New York, NY10010 \newline
and \newline
Department of Mathematical Sciences, \newline
Ritsumeikan University \newline
1-1-1 Nojihigashi, Kusatsu, Shiga, 525-8577, Japan
}
\email{tai-ho.wang@baruch.cuny.edu}

\keywords{Automatic market making, Decentralized exchange, Decentralized finance}

\begin{document}

\begin{abstract}
This paper develops a rigorous mathematical framework for analyzing Concentrated Liquidity Market Makers (CLMMs) in Decentralized Finance (DeFi) within a continuous-time setting. We model the evolution of liquidity profiles as measure-valued processes and characterize their dynamics under continuous trading. Our analysis encompasses two critical aspects of CLMMs: the mechanics of concentrated liquidity provision and the strategic behavior of arbitrageurs. We examine three distinct arbitrage models—myopic, finite-horizon, and infinite-horizon with discounted and ergodic controls—and derive closed-form solutions for optimal arbitrage strategies under each scenario. Importantly, we demonstrate that the presence of trading fees fundamentally constrains the admissible price processes, as the inclusion of fees precludes the existence of diffusion terms in the price process to avoid infinite fee generation. This finding has significant implications for CLMM design and market efficiency.
\end{abstract}

\maketitle

\input{Intro}
\input{CFMM}
\input{CLMM}

\input{Arbitrage}
\input{Conclusion}

\bibliographystyle{alpha}
\bibliography{Reference}
\end{document}

%% file: Intro.tex
\section{Introduction}
Decentralized Finance (\textbf{DeFi}) has revolutionized the financial landscape by enabling trustless, peer-to-peer transactions without intermediaries \cite{harvey2021defi, Gobet2023DecentralizedFB}. At the heart of this revolution are Automated Market Makers (\textbf{AMMs}) \cite{Capponi2021AdoptionDB, Loesch2022QuantitativeFinance}, which facilitate efficient and permissionless token exchanges. Unlike traditional limit order books (\textbf{LOBs}) that match discrete bids and asks \cite{Gould2013LimitOB}, AMMs employ mathematical formulas to determine prices based on asset quantities within liquidity pools, ensuring continuous liquidity availability and instant trade execution.

The introduction of Concentrated Liquidity Market Makers (\textbf{CLMMs}) by Uniswap V3 \cite{Adams2021UniswapV3} marked a paradigm shift in AMM design. Unlike traditional Constant Function Market Makers (\textbf{CFMMs}) \cite{Angeris2020ImprovedPO, angeris2022when, Angeris2023GeometryCFMM}, where liquidity is distributed uniformly across the entire price range according to a fixed bonding curve, CLMMs empower liquidity providers (\textbf{LPs}) to concentrate their capital within specific price intervals. This fundamental change shifts control of the market's bonding curve from a rigid, protocol-defined function to an emergent structure determined collectively by LP decisions. By enabling market participants to actively shape the liquidity distribution, CLMMs address a key inefficiency of CFMMs, where substantial capital may sit idle in rarely-traded price ranges. The resulting dynamic liquidity profile more closely resembles traditional limit order books \cite{Milionis2023complexityapproximation}, combining the capital efficiency of order book markets with the automation and permissionless nature of AMMs.

Recent literature has made significant strides in understanding CLMM dynamics. Heimbach et al.\ \cite{heimbach2023risks} provide a comprehensive analysis of the risks and returns associated with concentrated liquidity provision. Echenim et al.\ \cite{echenim2023thorough} develop precise models of CLMM mechanics using ``liquidity curves'' to analyze strategies and calculate LP fee earnings. Urusov et al.\ \cite{urusov2024backtesting} contribute specialized backtesting frameworks for evaluating CLMM performance.

Milionis et al.\ \cite{Milionis2022AutomatedMM, Milionis2022Quantifying} explored the theoretical foundations of AMM dynamics under myopic arbitrage, deriving a Black-Scholes-type formula, termed \textit{``loss-versus-rebalancing''} (\textbf{LVR}), to quantify LPs' adverse selection costs. Their subsequent work \cite{Milionis2023AutomatedMM} extended this analysis to incorporate trading fees and discrete block generation times. Fukasawa et al. \cite{fukasawa2023weighted, fukasawa2023modelfree} discussed model-free hedging strategies in the context of Geometric Mean Market Makers (\textbf{G3Ms}). Further research has broadened the analytical framework, with Bergault et al.\ \cite{bergault2024priceaware} and Cartea et al.\ \cite{cartea2023automated} exploring innovative AMM designs based on stochastic control theory.

Despite these advances in understanding CLMMs, a unified mathematical framework that fully captures the dynamic interplay between concentrated liquidity provision, price movements, and trading fees, comparable to those developed for LOBs \cite{Cont2010StochasticModel, Cont2013PriceDynamics, Cont2021StochasticPartial, ContDegondXuan2023}, has remained elusive. This paper addresses this gap through three main contributions:
\begin{enumerate}
\item We introduce a framework for modeling liquidity profiles as measure-valued processes (see Equations \eqref{eqn:CLMM_LP_xy} and \eqref{eqn:CLMM_orderflow_xy}). This formulation precisely characterizes how concentrated liquidity affects market behavior and trading outcomes.
\item Building upon the joint work with C.-Y. Lee \ \cite{lee2024growth}, we employ stochastic analysis and control theory to analyze liquidity profile dynamics under arbitrage. This analysis reveals that the presence of trading fees imposes fundamental constraints on admissible price processes, specifically precluding diffusion terms to avoid infinite fee generation.
\item We derive closed-form solutions for optimal arbitrage strategies under three distinct scenarios: myopic arbitrage, finite-horizon optimization, and infinite-horizon optimization with discounted or ergodic control. These solutions provide valuable insights into how rational actors interact with CLMMs and influence price discovery and market efficiency.
\end{enumerate}

Our findings have significant implications for CLMM design and management, offering guidance on optimizing fee structures and liquidity incentives to enhance market efficiency. The framework we develop bridges classical financial theory with emerging DeFi mechanisms, providing valuable insights for both researchers and practitioners in the field.

\subsection*{Outline}
The remainder of this paper is structured as follows: Section \ref{section:CFMM} examines the fundamentals of CFMMs, establishing the mathematical groundwork. Section \ref{section:CLMM} introduces liquidity profiles in CLMMs and develops the measure-theoretic framework for analyzing their properties. Section \ref{section:Arbitrage} presents our main results on CLMM dynamics, including the characterization of optimal arbitrage strategies and their implications for price processes. Finally, Section \ref{section:Conclusion} concludes with a discussion of practical implications and directions for future research.

%% file: CFMM.tex
\section{Automated Market Maker Fundamentals} \label{section:CFMM}
This section provides a concise review of Constant Function Market Makers. Readers already familiar with this topic can proceed to the next section.

\subsection{Constant Function Market Makers}
\textit{Constant Function Market Makers} (\textbf{CFMMs}) \cite{Angeris2020ImprovedPO, angeris2022when, Angeris2023GeometryCFMM} is a fundamental class of AMMs that operate based on a bonding function. These systems facilitate decentralized trading by enabling interactions between two primary categories of participants:
\begin{itemize}
\item \textbf{Liquidity Providers (LPs):} Users who deposit assets into the CFMM and receive liquidity tokens representing their pool share.
\item \textbf{Liquidity Takers (Traders):} Users who exchange one asset for another using the liquidity provided in the pool.
\end{itemize}

\subsubsection{Trading Mechanism}
The core mechanism of a CFMM relies on its bonding function, which represents a level set:
$$
f(x,y) = \ell,
$$
where
\begin{itemize}
    \item $x, y \in \mathbb{R}_+$ represent the quantities of the risk asset and numéraire, respectively;
    \item $\ell \in \mathbb{R}_+$ denotes the liquidity amount of the pool;
    \item $f: \mathbb{R}_+ \times \mathbb{R}_+ \to \mathbb{R}_+$ is a twice continuously differentiable function.
\end{itemize}

A trade, represented by quantity changes $(\Delta x, \Delta y)$, is considered valid if and only if it preserves the bonding function:
$$
f(x+\Delta x, y+\Delta y) = f(x, y).
$$
This constraint ensures that the liquidity amount remains constant throughout each trade, maintaining the fundamental relationship between asset quantities.

For a CFMM to support viable trading, the bonding function must satisfy several key properties:
\begin{itemize}
    \item Monotonicity: The implicit function $y(x)$ derived from $f(x, y) = \ell$ must be strictly decreasing.
    \item Convexity: The trading curve must be convex (see Equation \eqref{eqn:price_impact}) to ensure price impact increases with trade size.
    \item Scaling: For any $\lambda>0$, $f(\lambda x,\lambda y) = \lambda \ell$. This ensures reserves $(x,y)$ scale linearly with the liquidity amount $\ell$.
\end{itemize}

\subsubsection{Price Formation}
The instantaneous price (or spot price) of the risk asset in terms of the numéraire is derived using the implicit function theorem:
\begin{equation} \label{eqn:inf_price}
P = -\frac{dy}{dx} = \frac{f_x}{f_y},
\end{equation}
where $f_x$ and $f_y$ are partial derivatives of $f$ with respect to $x$ and $y$. This ratio of partial derivatives represents the slope of the tangent line to the trading curve at the point $(x,y)$.
For a given liquidity level $\ell$, we can express $x$ as a function of $y$:
$$
x = \phi(y, \ell)
$$
where $\phi$ is obtained by solving $f(x, y) = \ell$ for $x$. This allows us to express price as a function of $y$ and $\ell$:
$$
P = \frac{f_x(\phi(y, \ell),y)}{f_y(\phi(y, \ell), y)}.
$$
The marginal price impact of a trade is given by
\begin{equation} \label{eqn:price_impact}
\frac{dP}{dx} = -\frac{d^2 y}{d x^2} = - \frac{f_{xx} f_y - f_x f_{xy}}{(f_y)^2}.
\end{equation}
This relationship, which is positive due to the convexity of the bonding function, demonstrates how larger trades move the price more significantly, creating a natural price discovery mechanism.

\subsubsection{Liquidity Provision}
Liquidity providers (LPs) play a crucial role in CFMMs by depositing assets into the pool. When LPs add or remove liquidity, they do so proportionally to the existing reserves. In exchange for providing liquidity, LPs receive liquidity tokens. These tokens represent their pool share and can be redeemed later for a corresponding portion of the pool's assets. The minting and redemption mechanisms ensure that the relative proportions of the assets in the pool remain consistent as liquidity is added or removed. The additivity of liquidity across multiple LPs allows for efficient tracking and management of liquidity provision.

\subsubsection{Trading Fees}
To incentivize users to provide liquidity, CFMMs incorporate a trading fee. This fee is typically denoted by a parameter $\gamma \in (0,1)$, where $1-\gamma$ represents the percentage of the trade value that is collected as a fee. Common values for $\gamma$ range from 99\% to 99.99\%. The trading fee modifies the trading mechanism as follows:
\begin{itemize}
    \item For buying the risk asset ($\Delta x > 0$):
    $$
    f(x + \gamma \Delta x, y + \Delta y) = f(x,y).
    $$
    \item For selling the risk asset ($\Delta x < 0$):
    $$
    f(x + \Delta x, y + \gamma \Delta y) = f(x,y).
    $$
\end{itemize}
The trading fee creates an infinitesimal bid-ask spread, which is the difference between the price at which traders can buy and sell an infinitesimal amount of the risk asset:
$$
P_{bid} = \gamma^{-1} P, \quad P_{ask} = \gamma P,
$$
where $P$ is the instantaneous price (mid-price) derived earlier. This creates a bid-ask spread of
$$
P_{bid} - P_{ask} = (\gamma^{-1} - \gamma) P.
$$
The accumulated fees in terms of each asset are
$$
\Delta F_x = \frac{1-\gamma}{\gamma} \Delta x^+, \quad
\Delta F_y = \frac{1-\gamma}{\gamma} \Delta y^+
$$
where $\Delta x^+$ and $\Delta y^+$ represent positive changes in respective asset quantities.

\subsubsection{Pool Valuation}
The total value $V$ of the liquidity pool, expressed in terms of the numéraire (asset $Y$), is given by
\begin{equation} \label{eqn:pool_value}
    V = Px + y,
\end{equation}
where $P$ is the instantaneous price. This formula allows for a straightforward calculation of the pool's overall worth.
The pool's value function exhibits several important properties:
\begin{enumerate}
    \item Homogeneity: $V(\lambda x, \lambda y) = \lambda V(x, y)$ for $\lambda > 0$.
    \item Concavity in price: $\frac{\partial^2 V}{\partial P^2} = \frac{dx}{dP} \leq 0$.
\end{enumerate}
For a liquidity provider owning a fraction $\alpha$ of the pool, their position value is
$$
V_{\alpha} = \alpha(Px + y).
$$

\subsubsection{Examples:G3Ms}
Geometric Mean Market Makers (G3Ms), popularized by Balancer Protocol \cite{Martinelli2019Balancer, Evans2021Liquidity, Angeris2021Analysis}, employ the following bonding function:
\begin{equation} \label{eqn:G3M_bonding}
f(x,y) := x^w y^{1-w} = \ell,
\end{equation}
where $w \in (0,1)$ represents the weight of asset $X$ in the pool. This function enforces a constant weighted geometric mean relationship between the quantities of the two assets, $x$ and $y$. The instantaneous price $P$ in a G3M is given by
\begin{equation} \label{eqn:G3M_price}
P = -\frac{dy}{dx} = \frac{w}{1-w} \frac{y}{x}.
\end{equation}
Using \eqref{eqn:G3M_bonding} and \eqref{eqn:G3M_price}, we can express the reserves $x$ and $y$ in terms of the liquidity $\ell$ and price $P$:
\begin{equation} \label{eqn:inventory_G3M}
x = \left(\frac{w}{1-w} \right)^{1-w} LP^{w-1}, \quad 
y = \left(\frac{1-w}{w} \right)^{w} \ell P^{w}.
\end{equation}
This allows us to write the value function for G3M liquidity providers as
\begin{equation} \label{eqn:G3M_value}
V(P, \ell) = Px + y = \frac{\ell P^w}{w^{w} (1-w)^{1-w}}.
\end{equation}
This formulation demonstrates that a G3M maintains a portfolio with a fixed weight $w$ in asset $X$ and $(1-w)$ in asset $Y$, effectively implementing an automated constant rebalancing strategy.

A special case of the G3M is the Constant Product Market Maker (CPMM), popularized by Uniswap V2 \cite{Adams2020UniswapV2, Angeris2021Analysis}. This corresponds to setting $w = \frac12$ in the G3M bonding function, resulting in the familiar constant product formula: $xy = \ell^2$

\begin{figure}[ht]
\centering
\includegraphics[width=0.7\linewidth]{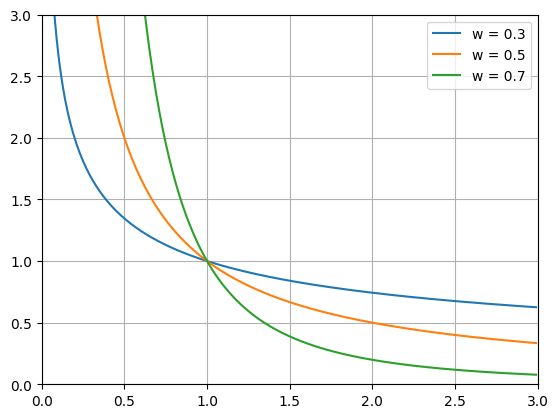}
\caption{G3M bonding curves for different values of $w$. These curves illustrate the relationship between the quantities of the two assets in the pool, as dictated by the G3M bonding function.}
\end{figure}

\subsection{Liquidity Provider Risks}
LPs in CFMMs face several inherent risks, with \textit{adverse selection} among the most significant. This risk arises when informed traders, particularly arbitrageurs, exploit price discrepancies between the CFMM and external markets. While such arbitrage activities help maintain market efficiency, they can substantially impact LPs' returns. This section provides a rigorous mathematical analysis of these risks, building upon and extending the framework presented in \cite{Milionis2022Quantifying, Milionis2022AutomatedMM, fukasawa2023weighted, fukasawa2023modelfree}.

\subsubsection{Impermanent loss}
\textit{Impermanent loss} (\textbf{IL}), also known as divergence loss, quantifies the opportunity cost of providing liquidity to a CFMM compared to simply holding the assets. We begin by establishing a formal framework for analyzing this phenomenon.

Consider an LP with initial holdings of $x_0 = x(P_0, \ell)$ units of asset $X$ and $y_0 = y(P_0, \ell)$ units of asset $Y$, where $P_0$ is the initial price of asset $X$ in terms of asset $Y$ within the CFMM pool.

If the LP holds these assets without trading, the value of their portfolio $H_t$ at time $t$ evolves according to
$$
dH_t = x(P_0) dP_t.
$$
Integrating this equation yields the portfolio value at time $t$:
$$
H_t = x_0 P_t + y_0 = x_0 P_0 + y_0 + x_0 \int_0^t dP_s.
$$
When providing liquidity to the CFMM pool, the position value $V_t$ at time $t$ is
$$
V_t = V(P_t) = x_tP_t + y_t = x(P_t) P_t + y(P_t).
$$
Applying Itô's formula to $V_t$ yields
$$
dV_t = V'(P_t) dP_t + \frac12 V''(P_t) d\langle P \rangle_t = x(P_t) dP_t + \frac12 V''(P_t) d\langle P \rangle_t,
$$
where $d\langle P \rangle$ denotes the quadratic variation of the price process. Here, we used the fact that
$$
V'(P_t) = \frac{d}{dP}\left(P x + y\right) = x + P\frac{dx}{dP} + \frac{dy}{dP} = x.
$$
The impermanent loss, ${\rm IL}_t:= H_t - V_t$, represents the difference in value between holding assets and providing liquidity:
$$
d{\rm IL}_t = dH_t - dV_t = \left\{x(P_0) - x(P_t)\right\} dP_t - \frac12 V''(P_t) d\langle P \rangle_t
$$
or equivalently, in integral form:
$$
{\rm IL}_t = H_t - V_t = \int_0^t \left\{x(P_0) - x(P_s)\right\} dP_s - \int_0^t\frac12 V''(P_s) d\langle P \rangle_s.
$$
Given that $V''(P) = x'(P) \leq 0$ (since the amount of asset $X$ held in the pool decreases as the price $P$ increases), and assuming that the price process $P_t$ is a martingale, we have
$$
\mathbb{E}[{\rm IL}_t] = -\mathbb{E}\left[\int_0^t\frac{1}{2} V''(P_s) d\langle P\rangle_s\right] \geq 0.
$$
This result demonstrates that, on average, an LP would have been better off holding the assets rather than providing liquidity.

\subsubsection{Loss-versus-Rebalance}
The quadratic variation term in the IL equation can be attributed to the concept of \textit{loss-versus-rebalance} (\textbf{LVR}).

Consider a self-financing trading strategy that continuously rebalances the portfolio to hold $x(P_t)$ units of asset $X$. Let $R_t$ denote the strategy's value at time $t$, with initial value $R_0 = x_0 P_0 + y_0$. The self-financing condition implies
$$
dR_t = x(P_t) dP_t.
$$
Define the loss-versus-rebalance as ${\rm LVR}_t := R_t - V_t$. Then,
$$
d{\rm LVR}_t := dR_t - dV_t = -\frac12 V''(P_t) d\langle P \rangle_t \geq 0.
$$
This inequality holds because $V''(P) = x'(P) \leq 0$.
The impermanent loss can be decomposed as
\begin{equation} \label{eqn:IL_decomp}
{\rm IL}_t = \int_0^t \left\{x(P_0) - x(P_s)\right\} dP_s + {\rm LVR}_t.
\end{equation}
This decomposition reveals two components:
\begin{enumerate}
    \item A martingale component that can be hedged through dynamic position adjustment.
    \item A non-negative drift component representing the loss-versus-rebalance.
\end{enumerate}

\subsubsection{Example: IL and LVR in G3Ms}
For G3Ms, applying Itô's formula to the value function \eqref{eqn:G3M_value} yields
$$
dV_t = \left(\frac{w}{1-w}\right)^{1-w} \ell P_t^{w-1} dP_t - \frac{w^{1-w}(1-w)^w}{2} \ell P_t^{w-2} d\langle P \rangle_t.
$$
In comparison, the value $R_t$ of a self-financing portfolio continuously rebalanced to hold $x_t$ units of the risky asset evolves according to
$$
dR_t = x_t dP_t = \left(\frac{w}{1-w}\right)^{1-w} \ell P_t^{w-1} dP_t
$$
The difference between these two processes represents the LVR:
$$
d{\rm LVR}_t = dR_t - dV_t = \frac{w^{1-w}(1-w)^w}{2} \ell P_t^{w-2} d\langle P \rangle_t.
$$
Using the IL decomposition \eqref{eqn:IL_decomp}, we obtain
$$
{\rm IL}_t = \int_0^t (x_0 - x_s) dP_s + \frac{w^{1-w}(1-w)^w}{2} \ell \int_0^t P_s^{w-2} d\langle P \rangle_s.
$$
This example demonstrates how the concepts of IL and LVR can be explicitly quantified for the specific case of a G3M.

%% file: CLMM.tex
\section{Concentrated Liquidity Market Makers} \label{section:CLMM}
Unlike traditional CFMMs, where liquidity is distributed uniformly along the entire price curve, \textit{Concentrated Liquidity Market Makers} (\textbf{CLMMs}) \cite{Adams2021UniswapV3} enable LPs to concentrate their capital within specific price ranges, substantially improving capital efficiency. This design allows LPs to optimize their returns by allocating liquidity where it is most likely to be used, reducing the amount of capital "idle" in price ranges with minimal trading activity.

\subsection{CLMM Architecture}
\subsubsection{Position Structure and Mechanics}
A CLMM LP position is defined by a pair $(\ell, [p_l, p_u])$, where $\ell$ represents the amount of liquidity and $0 < p_l < p_u < \infty$ defines the chosen price range for liquidity provision. The required reserves depend on the current pool price $P$:
\begin{enumerate}
    \item If $P \in [p_l, p_u]$, the LP deposits $\ell \left(\frac1{\sqrt P} - \frac1{\sqrt p_u} \right)$ of token $X$ and $\ell \left(\frac1{\sqrt P} - \frac1{\sqrt p_u} \right)$ of token $Y$ into the pool.
    \item If $P > p_u$, LP deposits only $\ell \left(\sqrt{p_u} - \sqrt{p_l}\right)$ of token $Y$ into the pool, as the price is above their specified range.
    \item If $P < p_l$, the LP deposits only $\ell \left(\frac1{\sqrt{p_l}} - \frac1{\sqrt{p_u}}\right)$ of token $X$ into the pool, as the price is below their specified range. 
\end{enumerate}

\subsubsection{Reserve Functions and Bonding Curves}
The complete reserve functions for any price $p$ can be expressed as:
\begin{align} \label{eqn:CLMM_xy}
& x(p) = \ell \left(\frac1{\sqrt p} - \frac1{\sqrt p_u} \right)^+ - \ell \left(\frac1{\sqrt p} - \frac1{\sqrt p_l} \right)^+, \\
& y(p) = \ell \left(\sqrt p - \sqrt p_l \right)^+ - \ell \left(\sqrt p - \sqrt p_u \right)^+. \notag
\end{align}
These reserve functions resemble the payoff function of a bull spread in a transformed price space (see, for instance, Equation \eqref{eqn:x_in_s}).

For $p \in [p_l, p_u]$, eliminating the parameter $p$ in Equation \eqref{eqn:CLMM_xy} yields the bonding curve
$$
\left( x + \frac{\ell}{\sqrt{p_u}} \right)^{1/2} \left(y + \ell\sqrt{p_l} \right)^{1/2} = \ell,
$$
or equivalently,
$$
\left( \frac{x}{\ell} + \frac1{\sqrt{p_u}}\right)\left( \frac{y}{\ell} + \sqrt{p_l} \right) = 1.
$$
Using Equation \eqref{eqn:pool_value} and \eqref{eqn:CLMM_xy}, the value of an LP position is given by
\begin{align}
    V(P)
    &= \ell P \left[ \left(\frac1{\sqrt P} - \frac1{\sqrt{p_u}} \right)^+ - \left(\frac1{\sqrt P} - \frac1{\sqrt{p_l}} \right)^+ + \left(\frac1{\sqrt P} - \frac{\sqrt{p_l}}{P} \right)^+ - \left(\frac1{\sqrt P} - \frac{\sqrt{p_u}}{P} \right)^+ \right] \notag \\
    &= \begin{cases}
        \ell P \left( \frac1{\sqrt{p_l}} - \frac1{\sqrt{p_u}} \right) &\text{ if } P < p_l, \\
        \ell P \left( \frac2{\sqrt{P}} - \frac{1}{\sqrt{p_u}} - \frac{\sqrt{p_l}}{P} \right) &\text{ if } p_l \leq P \leq p_u, \\
        \ell \left( \sqrt{p_u} - \sqrt{p_l} \right) &\text{ if } P > p_u,
    \end{cases} \label{eqn:CLMM_LP_value}
\end{align}
where $(z)^+ = \max(z,0)$ denotes the positive part of $z$.

\subsubsection{Relationship to Traditional CFMMs}
Uniswap V3 introduced the concept of concentrated liquidity, allowing LPs to provide liquidity within specific price ranges. This effectively provides liquidity on a "portion" of the Uniswap V2 bonding curve (see Figure \ref{fig:CLMM_bonding}).

Key distinctions between CLMMs and traditional CFMMs include:
\begin{itemize}
    \item In a CPMM, the bonding curve never intersects the axes, implying that liquidity is provided across the entire price range $[0,\infty)$.
    \item In a CLMM, if a single LP provides liquidity within the range $[p_l, p_u]$, the bonding curve intersects the $x$-axis at $x^* = \ell \left(\frac1{\sqrt p_l} - \frac1{\sqrt p_u}\right)$ and the $y$-axis at $y^*=L\left({\sqrt p_u} - {\sqrt p_l}\right)$. This indicates that no swaps are possible beyond the LP's specified liquidity range.
\end{itemize}

\begin{figure}[ht]
\centering
\includegraphics[width=0.7\linewidth]{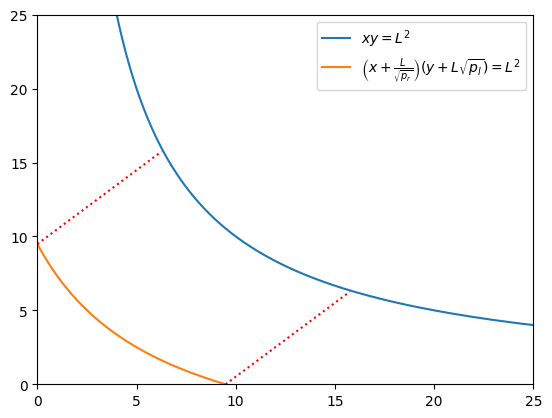}
\caption{CLMM bonding curve with $L=10$, $p_l=0.4$, $p_u=2.5$}
\label{fig:CLMM_bonding}
\end{figure}

\subsubsection{Connection to Covered Call Strategies}
Providing liquidity in a CLMM can be likened to a covered call strategy in traditional finance. Let $p_m = \sqrt{p_l p_u}$ be the geometric mean of the lower and upper bounds of the LP's price range. We can then express these bounds as 
$$
p_l = \frac{p_m}{r}, \quad p_u = r p_m,
$$
where $r = \frac{p_u}{p_l} > 1$ represents the relative width of the price range. With this notation, the LP's value function \eqref{eqn:CLMM_LP_value} can be rewritten as
\begin{align*}
    V(k)
    &= \ell p_m^{\frac12} \left[ \left( \sqrt k - \frac{k}{\sqrt r} \right)^+ - \left( \sqrt k - k \sqrt r \right)^+ + \left( \sqrt k - \frac{1}{\sqrt r} \right)^+ - \left( \sqrt k - \sqrt r \right)^+ \right] \notag \\
    &= \begin{cases}
        \ell \sqrt{p_m} \left( \sqrt r - \frac1{ \sqrt r} \right) k &\text{ if } k < r^{-1}, \\
        \ell \sqrt{p_m} \left( 2 - \frac{\sqrt{k}}{\sqrt r} - \frac1{\sqrt{kr}} \right) \sqrt{k} &\text{ if } r^{-1} \leq k \leq r, \\
        \ell \sqrt{p_m} \left( \sqrt r - \frac1{\sqrt r} \right) &\text{ if } k > r,
    \end{cases}
\end{align*}
where $k = \frac{P}{p_m}$ is the ratio of the current price $P$ to the center price $p_m$. This function exhibits the following properties:
\begin{itemize}
    \item For any $k \geq 0$,  
    \begin{equation} \label{eqn:V-vs-covered-call}
        V(k) \leq \ell \sqrt{p_m}\left(\sqrt r - \frac1{\sqrt r}\right) \left\{k - (k - 1)^+\right\};
    \end{equation}
    \item As the price range narrows ($r \to 1$, implying $p_l \to p_u$), the LP's value function converges to the right-hand side of inequality \eqref{eqn:V-vs-covered-call}.
\end{itemize}
These properties illustrate that the LP's value function resembles the payoff of a portfolio of $\ell \sqrt{p_m}\left(\sqrt r - \frac1{\sqrt r}\right)$ covered call options with a strike price of 1 in the $k$-domain (see Figure \ref{fig:CLMM_value}). This analogy highlights the similarity between concentrated liquidity provision and holding covered calls, where the LP's potential fees earned from trading activity are akin to the premiums received from selling call options.

\begin{figure}[ht]
\centering
\includegraphics[width=0.7\linewidth]{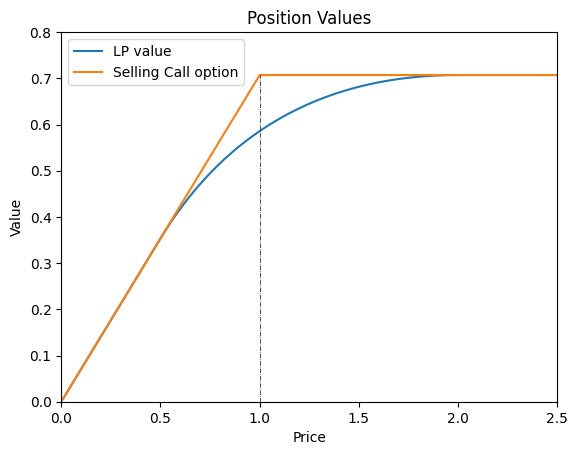}
\caption{CLMM LP value with $L=1$, $p_l = 0.5$, $p_u = 2$}
\label{fig:CLMM_value}
\end{figure}

\subsection{Liquidity Profile}
In Uniswap V3, liquidity providers gain the ability to concentrate their liquidity within specific price ranges, leading to a piecewise constant or step-function representation of liquidity within the pool. To formalize this concept, we introduce the notion of a liquidity profile, denoted by $\ell = \ell(P)$, which describes the distribution of liquidity as a function of the price $P$ over the entire price spectrum $(0,\infty)$. This profile provides a comprehensive view of how liquidity is allocated across different price levels, capturing the essence of concentrated liquidity provision.

\subsubsection{Single Position Analysis} \label{section:single_CLMM}
To establish a foundation for understanding liquidity profiles, we begin by analyzing a single CLMM position with liquidity $\ell$ concentrated within a specific price range $[p_l,p_u]$. We utilize a change of variables approach to simplify the analysis of the reserve functions.

For the reserve function $x(p)$, representing the quantity of asset $X$ held in the pool, we employ the substitution $s = \frac{1}{\sqrt{p}}$. This transformation yields
\begin{equation} \label{eqn:x_in_s}
x\left(\frac{1}{s^2}\right) = \ell \left\{(s - s_r)^+ - (s - s_l)^+ \right\},
\end{equation}
where $s_r = \frac{1}{\sqrt{p_r}}$ and $s_l = \frac{1}{\sqrt{p_l}}$. Notably, the right-hand side of Equation \eqref{eqn:x_in_s} corresponds to the payoff function of a long position in $\ell$ bull spreads with strikes  $s_r$ and $s_l$ (where $s_r < s_l$) in the $s$-domain. This connection to option payoffs provides an intuitive understanding of the reserve function's behavior.

Equation \eqref{eqn:x_in_s} can be equivalently represented using a Lebesgue integral:
$$
x\left(\frac{1}{s^2}\right) = \int_{\mathbb{R}^+} (s - k)^+ d\tilde{\ell}(k) = \int_{[0, s]} (s - k) d\tilde{\ell}(k),
$$
where $d\tL$ is a signed measure defined as the sum of Dirac delta functions:
\begin{equation} \label{eqn:single-Dirac}
    d\tilde{\ell}(k) = \ell \{\delta(s_r) - \delta(s_l)\}.
\end{equation}
Let $\tL(k) = \int_{[0,k]} d\tL(r)$ be the cumulative distribution function of $d\tL(r)$.  Assuming $\tL(0) = 0$ and integrating by parts, we obtain
$$
x\left(\frac1{s^2}\right) = (s - k) \tilde{\ell}(k)|_0^s + \int_{[0, s]} \tilde{\ell}(k) dk = \int_{[0, s]} \tilde{\ell}(k) dk.
$$
Transforming back to the price space with the substitution $k = \frac1{\sqrt q}$, we arrive at
\begin{equation} \label{eqn:x_dist_L}
x(p) = x\left(\frac1{s^2}\right) = \int_{[0,s]} \tL(k) dk = \frac12 \int_p^{\infty} \ell(q) q^{-\frac32} dq,
\end{equation}
where $\ell(p) := \tL\left( \frac1{\sqrt p} \right)$. 

Similarly, for the reserve function $y(p)$, representing the quantity of asset $Y$, we substitute $t = \sqrt{p}$ into \eqref{eqn:CLMM_xy} to obtain
$$
y(p) = y(t^2) = \ell \left[(t - t_l)^+ - (t - t_u)^+\right],
$$
where $t_l = \sqrt{p_l}$, $t_u=\sqrt{p_u}$. An analogous analysis leads to
\begin{equation} \label{eqn:y_dist_L}
y(p) = \int_{[0, t]} \tL(k) dk = \frac12 \int_0^{p} \ell(q) q^{-\frac12} dq,
\end{equation}
where $\ell(p) := \tL(\sqrt p)$.

It is crucial to note that the measure and its associated distribution function in Equation \eqref{eqn:single-Dirac} are linear in the liquidity parameter $\ell$. This linearity allows us to extend the above analysis and equalities to general $\sigma$-finite signed measures through a standard limiting process, providing a framework for analyzing more complex liquidity profiles.

\subsubsection{Liquidity as a Distribution}
Building on the analysis in Section \ref{section:single_CLMM}, we generalize the concept of liquidity to be a function of the instantaneous price $P$, represented by the \textit{liquidity profile} $\ell(P)$. This function describes the distribution of liquidity across different price levels, capturing the essence of concentrated liquidity provision.

From a measure-theoretic perspective, $d\ell(p)$ can be interpreted as a $\sigma$-finite signed measure on the interval $[0, \infty)$, with $\ell(p)$ serving as its cumulative distribution function (CDF). This interpretation allows us to leverage the tools of measure theory to analyze the properties of liquidity profiles and their impact on CLMM behavior.

Using Equations \eqref{eqn:x_dist_L} and \eqref{eqn:y_dist_L}, we express the pool reserves in terms of the instantaneous price $P$:
\begin{equation} \label{eqn:CLMM_LP_xy}
x(P) = \frac12 \int_P^\infty \ell(p) p^{-\frac32} dp, \quad
y(P) = \frac12 \int_0^P \ell(p) p^{-\frac12} dp.
\end{equation}
These equations establish a direct relationship between the reserves and the liquidity profile, highlighting how the distribution of liquidity influences the quantities of assets held in the pool. Differentiating \eqref{eqn:CLMM_LP_xy} with respect to $P$, we confirm that
$$
\frac{d y}{dx} = \frac{\frac{dy}{dP}}{\frac{dx}{dP}} = -P < 0.
$$
This result verifies that, as in standard CFMMs, the instantaneous price $P$ corresponds to the negative slope of the tangent to the reserve curve. Further analysis of the shape of this curve reveals that
$$
\frac{d^2 y}{dx^2} = \frac{\frac1{dP}\left(\frac{dy}{dx}\right)}{\frac{dx}{dP}} = \frac{-1}{-\frac{\ell(P)}{2P^{\frac32}}} = \frac{2P^{\frac32}}{\ell(P)}> 0.
$$
These findings indicate that the reserve curve is both decreasing and convex, ensuring that the price impact of trades increases with trade size, a crucial property for market stability.

Furthermore, assume the boundary conditions 
$$
\lim_{p \to \infty} \ell(p) p^{-\frac12} = 0 \quad \text{and} \quad \lim_{p \to 0} \ell(p) p^{\frac12} = 0,
$$
and applying integration by parts to Equation \eqref{eqn:CLMM_LP_xy}, we obtain
\begin{align*}
x(P) &= \frac{\ell(P)}{\sqrt P} + \int_P^\infty \frac1{\sqrt p} d\ell(p), \\
y(P) &= \ell(P) \sqrt P - \int_0^P \sqrt p d\ell(p).
\end{align*}
These expressions provide an alternative representation of the pool reserves in terms of the liquidity profile and its integral.

\subsubsection{Examples and Empirical Data}
\begin{example}[Uniform Liquidity]
The simplest case is a uniform liquidity distribution, where $\ell(P) = \ell$ for all $P \in (0,\infty)$. This corresponds to the CPMM, where liquidity is evenly distributed across all price levels. In this case, Equations \eqref{eqn:CLMM_LP_xy} reduce to the familiar constant product formula $xy=L^2$.

\begin{figure}[ht]
\centering
\includegraphics[width=0.8\linewidth]{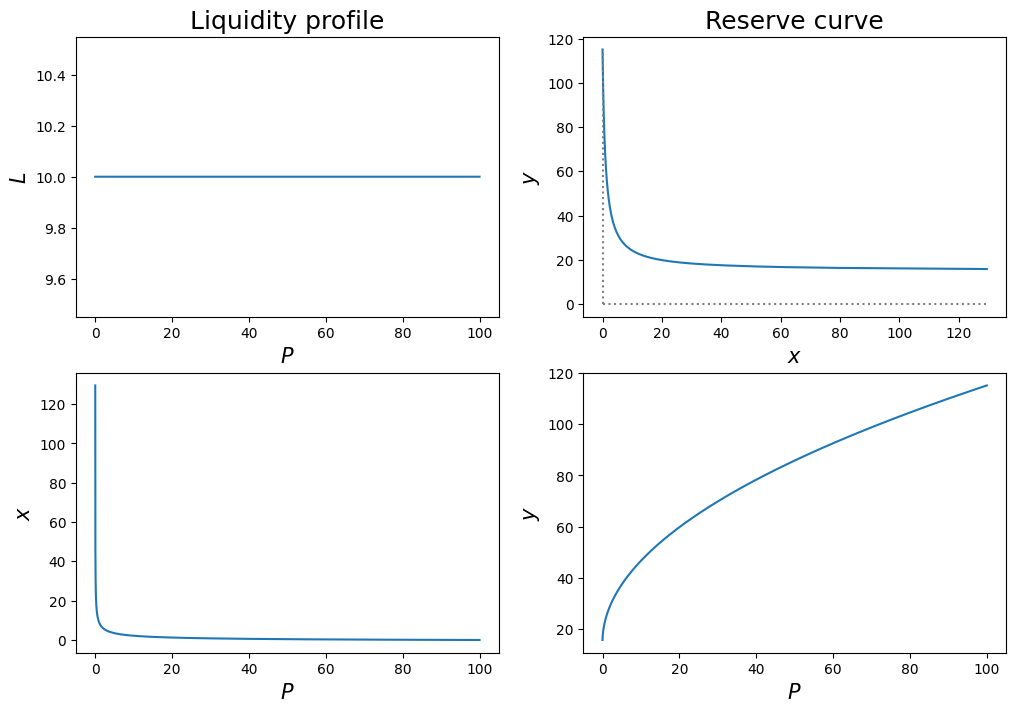}
\caption{Liquidity profile for Uniswap V2, illustrating a constant liquidity distribution across all prices.}
\end{figure}
\end{example}

\begin{example}[Two Disjoint Positions]
Consider a liquidity profile constructed by combining two distinct LP positions:
$$
\ell(p) = k_1 \1_{[a, b)}(p) + k_2 \1_{[b, c)}(p),
$$
where $k_1$ and $k_2$ are positive constants and $\1_{[a,b)}$ denotes the indicator function of the interval $[a,b)$. This profile represents two concentrated liquidity positions with different intensities within distinct, non-overlapping price ranges. Calculating the integrals in \eqref{eqn:CLMM_LP_xy} for the different price intervals yields:
\begin{enumerate}
    \item For $P \geq c$, 
    \begin{align*}
        x = 0, \quad y = -k_1 \sqrt a - (k_2 - k_1)\sqrt b + k_2\sqrt c = k_1\left(\sqrt b - \sqrt a\right) + k_2\left(\sqrt c - \sqrt b\right).
    \end{align*}
    \item For $b \leq P < c$, 
    \begin{align*}
        x &= \frac{\ell(P)}{\sqrt P} + \int_P^\infty \frac1{\sqrt p} d\ell(p) = \frac{k_2}{\sqrt P} - \frac{k_2}{\sqrt c} = k_2 \left(\frac1{\sqrt P} - \frac1{\sqrt c}\right), \\
        y &= \sqrt P \ell(P) - \int_0^P \sqrt p d\ell(p) = k_2 \sqrt P - k_1 \sqrt a - (k_2 - k_1)\sqrt b \\
        &= k_1 \left(\sqrt b - \sqrt a\right) + k_2 \left(\sqrt P - \sqrt b\right).
    \end{align*}
    \item For $a \leq P < b$,
    \begin{align*}
        x &= \frac{\ell(P)}{\sqrt P} + \int_P^\infty \frac1{\sqrt p}d\ell(p) = \frac{k_1}{\sqrt P} + (k_2 - k_1) \left(\frac1{\sqrt b} - \frac1{\sqrt c}\right) - \frac{k_2}{\sqrt c} \\
        &= k_1 \left(\frac1{\sqrt P} - \frac1{\sqrt b}\right) + k_2 \left(\frac1{\sqrt b} - \frac1{\sqrt c}\right), \\
        y &= \sqrt P \ell(P) - \int_0^P \sqrt p d\ell(p) = k_2 \sqrt P - k_1 \sqrt a - (k_2 - k_1)\sqrt b \\
        &= k_1 \left(\sqrt b - \sqrt a\right) + k_2 \left(\sqrt P - \sqrt b\right).
    \end{align*}
    \item For $P < a$,
    \begin{align*}
        x =  \frac{\ell(P)}{\sqrt P} + \int_P^\infty \frac1{\sqrt p} d\ell(p) = \frac{k_1}{\sqrt a} + \frac{k_2 - k_1}{\sqrt b} - \frac{k_2}{\sqrt c}, \quad
        y = 0.
    \end{align*}
\end{enumerate}
\begin{figure}[ht]
\centering
\includegraphics[width=0.8\linewidth]{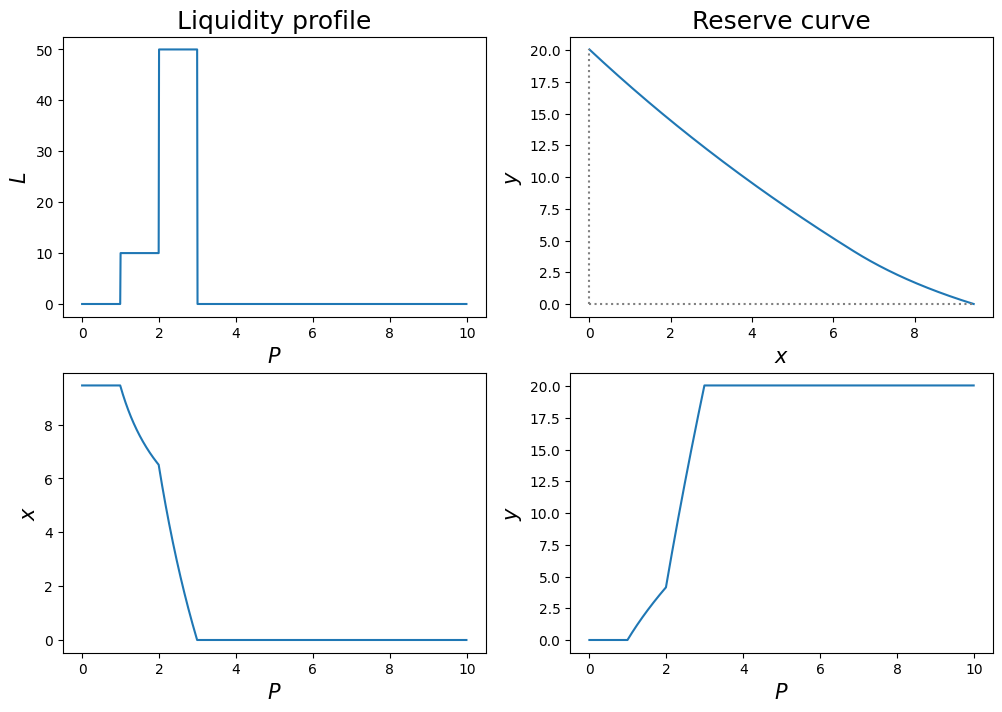}
\caption{Liquidity profile with two concentrated liquidity positions, illustrating the varying liquidity levels across different price ranges.}
\end{figure}
\end{example}
\begin{example}[Continuous Distribution]
Assume the liquidity profile is given by a $\chi^2$ distribution with three degrees of freedom. In this case, the expressions in \eqref{eqn:CLMM_LP_xy} do not have closed-form analytical solutions. To illustrate the behavior of the liquidity profile and the corresponding pool reserves, Figure \ref{fig:lp_chi2} presents numerical results parameterized by pool price $P$.
\begin{figure}[ht]
\centering
\includegraphics[width=0.8\linewidth]{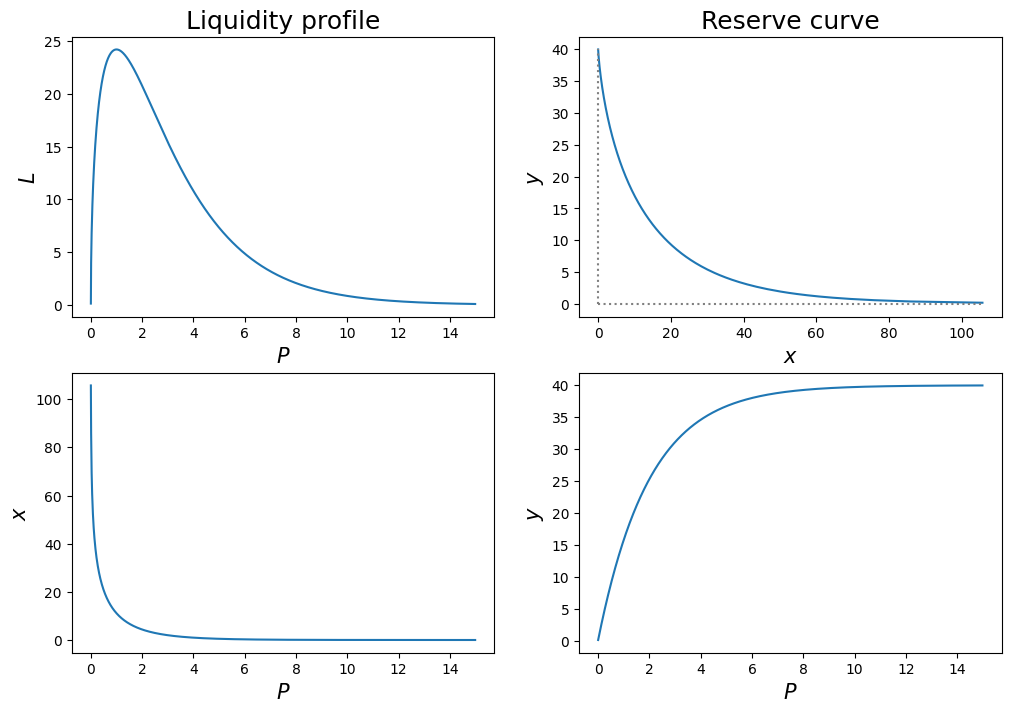}
\caption{Liquidity profile following a $\chi^2$-distribution, demonstrating a continuous and non-uniform liquidity distribution.}
\label{fig:lp_chi2}
\end{figure}
\end{example}

\begin{example}[Real-World Liquidity Profiles]
Figure \ref{fig:empirical-profiles} illustrates empirical liquidity profiles observed on Uniswap V3. The top 3D plot shows the evolution of liquidity over time, with "tick" representing discrete price levels and the vertical axis indicating liquidity depth (on a logarithmic scale). The bottom four plots provide snapshots of liquidity profiles at different points in time, highlighting the varying concentration of liquidity around the spot price. Notably, the liquidity profiles tend to be highly concentrated around the current spot price, reflecting the active management of liquidity positions by liquidity providers.

\begin{figure}
    \centering
    \begin{subfigure}[b]{0.8\textwidth}
        \centering
        \includegraphics[width=\textwidth]{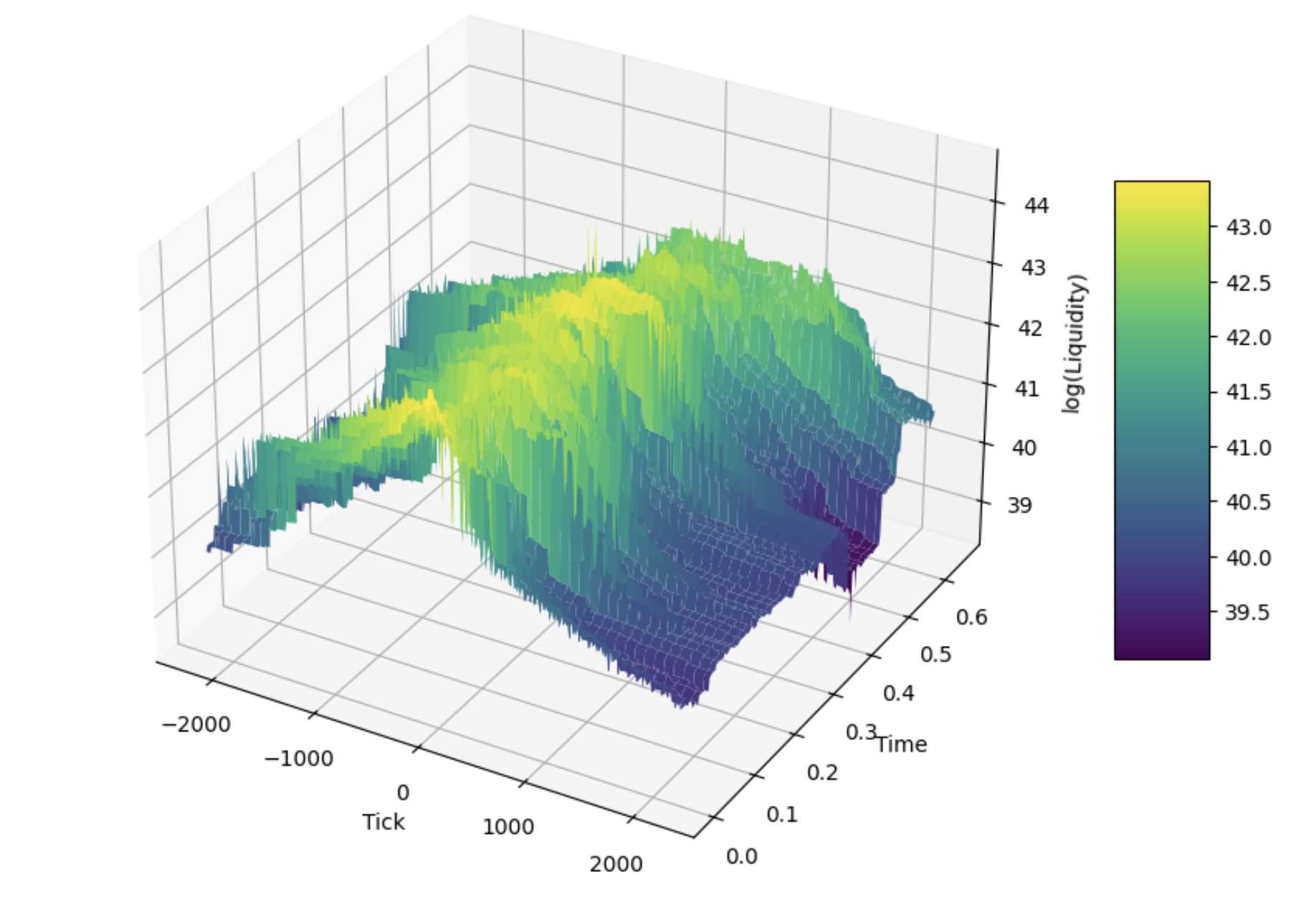}
        \caption{3D Liquidity Evolution}
    \end{subfigure}
    
    \vspace{1em}
    
    \begin{subfigure}[b]{0.43\textwidth}
        \centering
        \includegraphics[width=\textwidth]{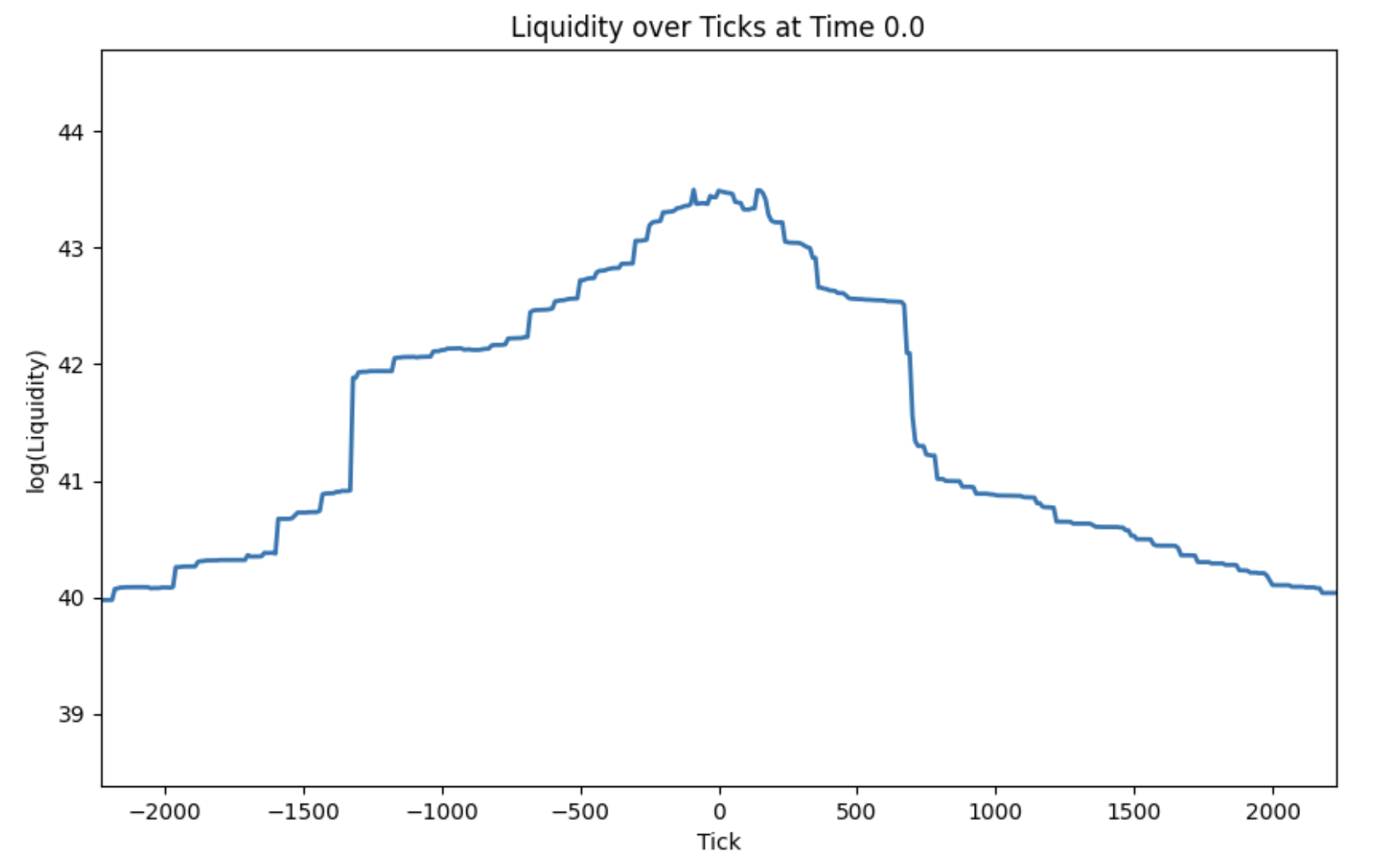}
        \caption{Snapshot 1}
    \end{subfigure}
    \hfill
    \begin{subfigure}[b]{0.43\textwidth}
        \centering
        \includegraphics[width=\textwidth]{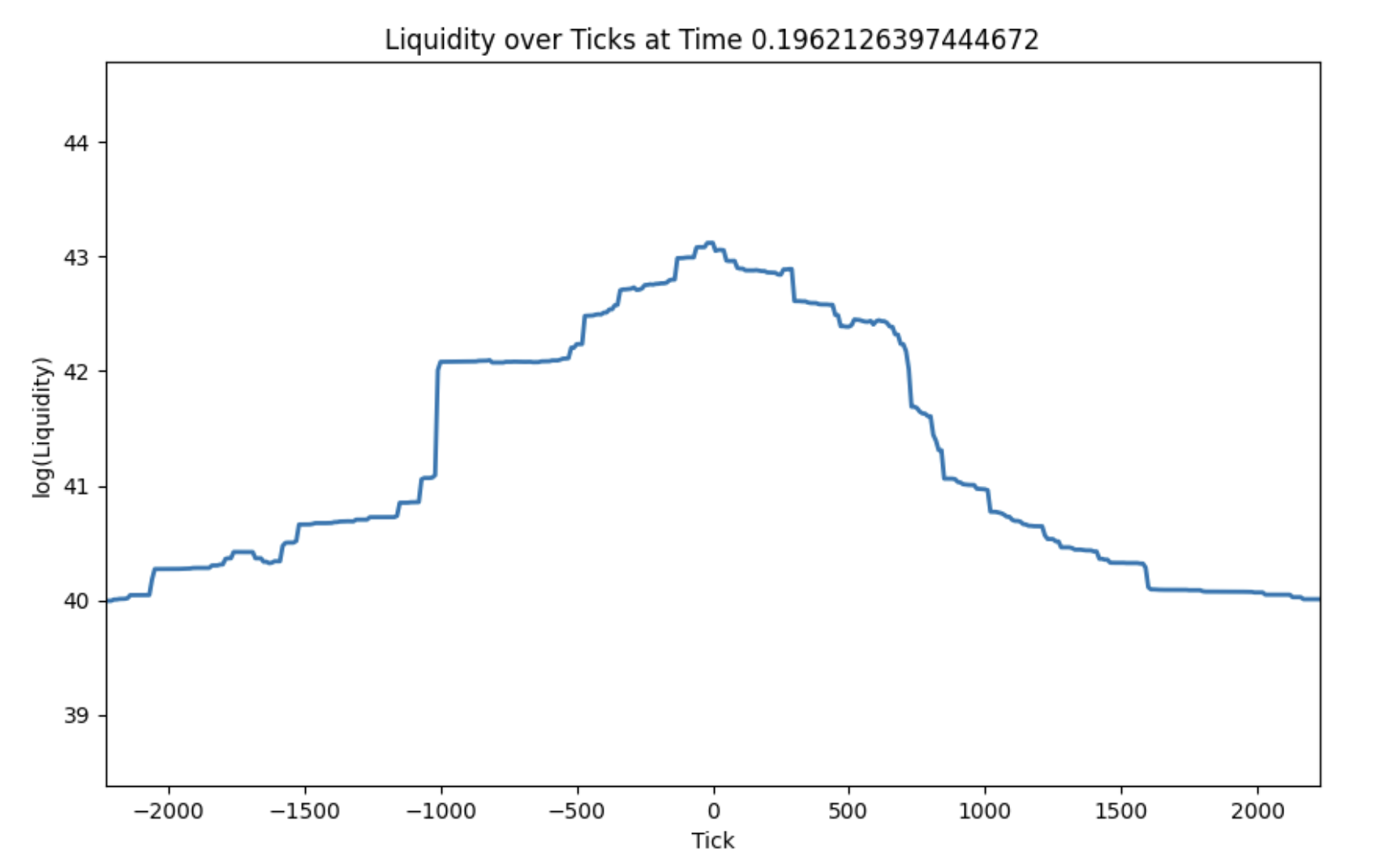}
        \caption{Snapshot 2}
    \end{subfigure}
    \hfill
    \begin{subfigure}[b]{0.43\textwidth}
        \centering
        \includegraphics[width=\textwidth]{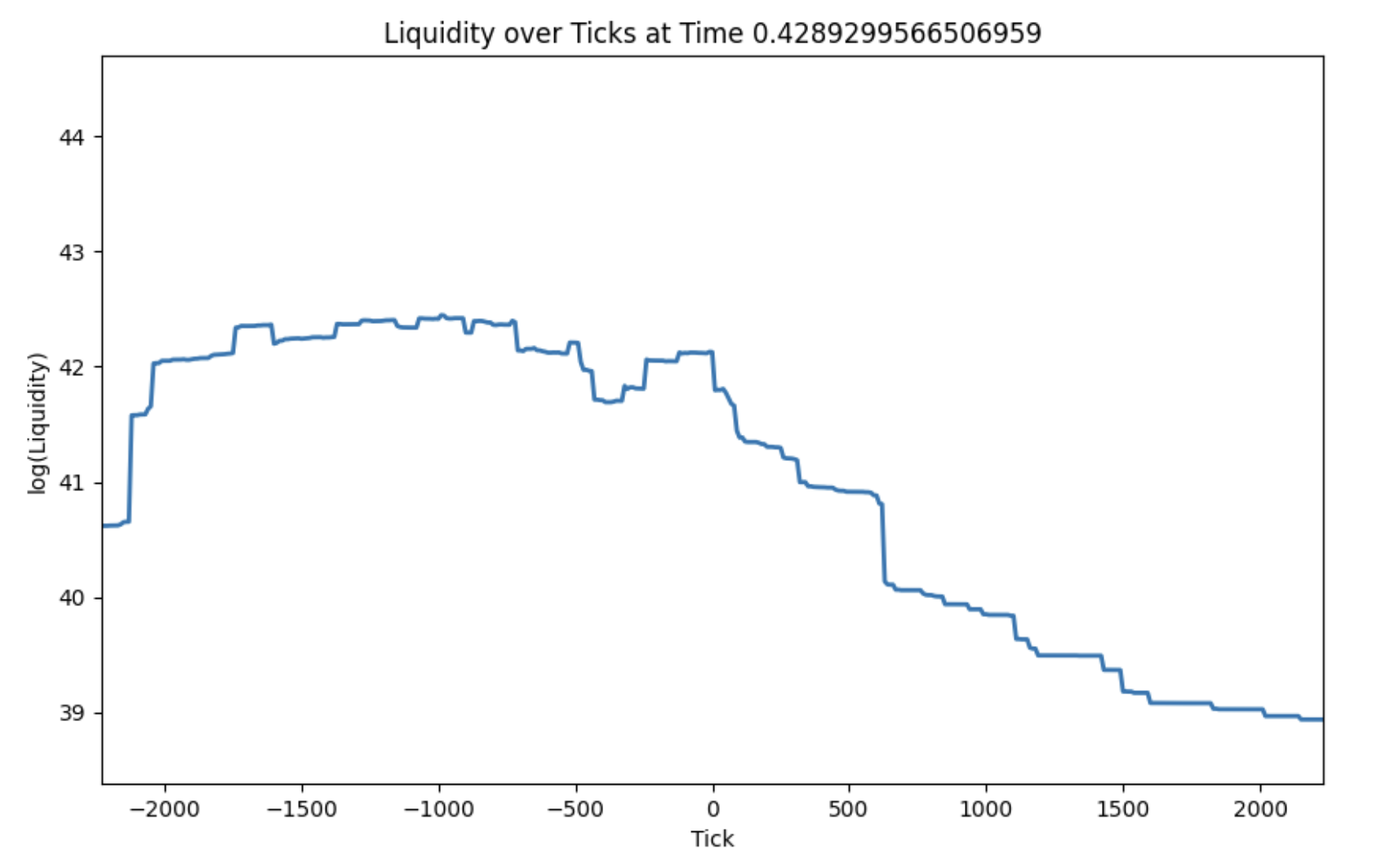}
        \caption{Snapshot 3}
    \end{subfigure}
    \hfill
    \begin{subfigure}[b]{0.43\textwidth}
        \centering
        \includegraphics[width=\textwidth]{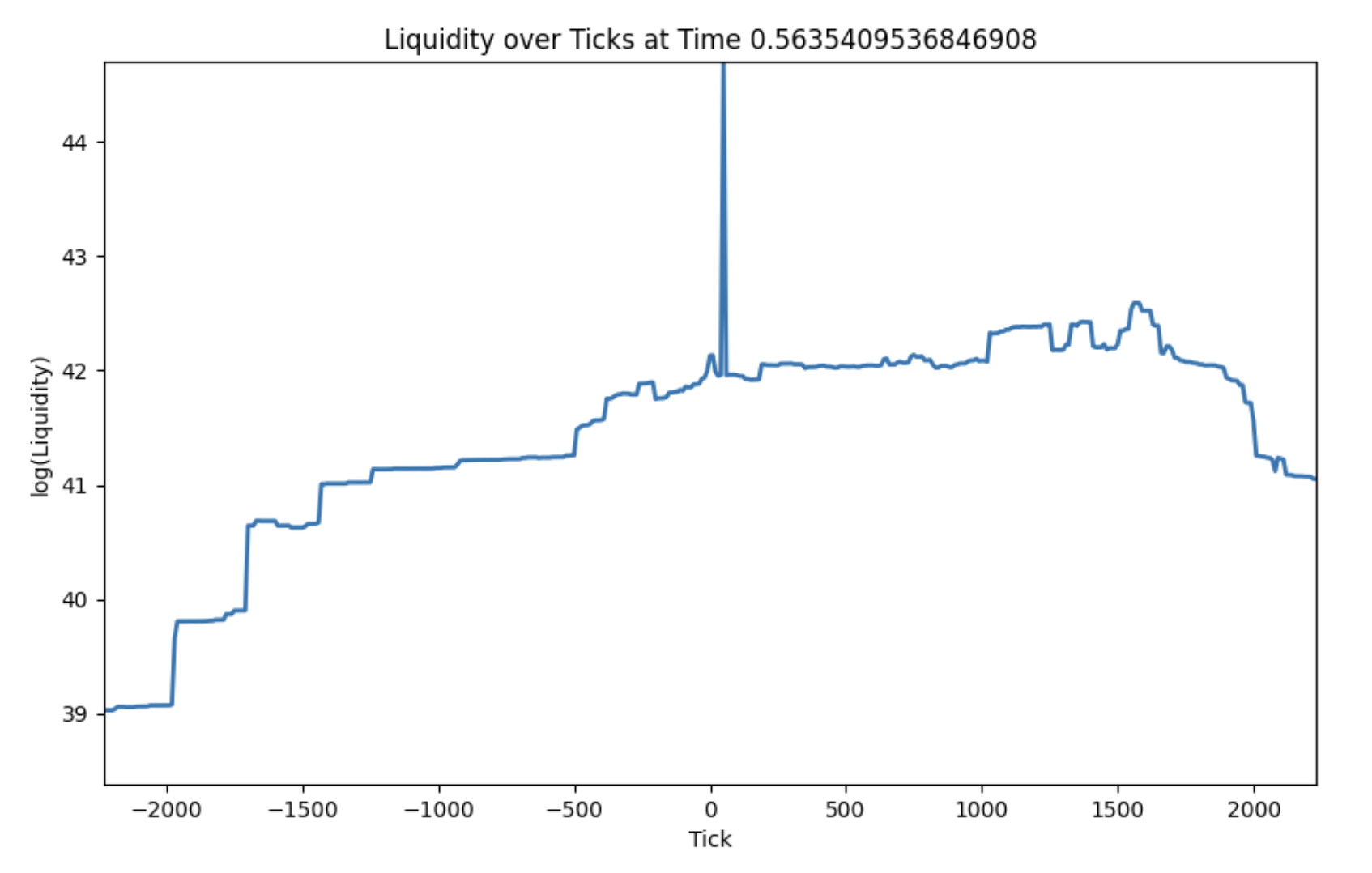}
        \caption{Snapshot 4}
    \end{subfigure}
    
    \caption[size=normalsize]{Empirical liquidity profiles from Uniswap V3. The 3D plot shows liquidity evolution over time and across "ticks" (discrete price levels). The 2D plots display snapshots of liquidity profiles at different times, with liquidity shown on a logarithmic scale. Data courtesy of Wun-Cing Liou and Jimmy Risk.}
    \label{fig:empirical-profiles}
\end{figure}
\end{example}

\subsubsection{IL and LVR for Static Liquidity Profile} \label{section:LVR_CLMM}
We now examine how IL and LVR evolve in the context of a liquidity profile $\ell(P)$. Applying Itô's formula to \eqref{eqn:CLMM_LP_xy}, we obtain the dynamics of the reserve asset $x_t$:
\begin{align*}
dx_t 
&= d\left\{\frac12 \int_{P_t}^\infty \ell(p)p^{-\frac32} dp \right\} \notag \\
&= \frac12 \int_{P_t}^\infty \frac{d\ell}{dt}(p)p^{-\frac32} dp - \frac12 \ell(P_t){P_t}^{-\frac32} dP_t + \frac12 \left\{-\frac12 P_t \ell'(P_t)+ \frac34 \ell(P_t)\right\} P_t^{-\frac52} d\inn{P}_t, 
\end{align*}
and $y_t$:
\begin{align*}
dy_t 
=& d\left\{\frac12 \int_0^{P_t} \ell(p)p^{-\frac12} dp \right\} \notag \\
=& \frac12 \int_0^{P_t} \frac{d\ell}{dt}(p)p^{-\frac12} dp + \frac12 \ell(P_t){P_t}^{-\frac12} dP_t + \frac12 \left\{\frac12 P_t \ell'(P_t) - \frac14 \ell(P_t)\right\} P_t^{-\frac32} d\inn{P}_t.
\end{align*}
Note that $\frac{d\ell}{dt} = 0$ since the liquidity profile is static. Using these expressions, we derive the evolution of the total pool value $V_t = P_t x_t + y_t$:
\begin{align*}
dV_t 
=& d(P_t x_t + y_t) = x_t dP_t + P_t dx_t + d\inn{x, P}_t + dy_t \notag \\
=& x_t dP_t - \frac12 \ell(P_t){P_t}^{-\frac12} dP_t + \frac12 \left\{-\frac12 P_t \ell'(P_t)+ \frac34 \ell(P_t)\right\} P_t^{-\frac32} d\inn{P}_t  \notag \\
&- \frac12 \ell(P_t){P_t}^{-\frac32} d\inn{P}_t + \frac12 \ell(P_t){P_t}^{-\frac12} dP_t + \frac12 \left\{\frac12 P_t \ell'(P_t) - \frac14 \ell(P_t)\right\} P_t^{-\frac32} d\inn{P}_t  \notag \\
=& x_t dP_t - \frac14 \ell(P_t) P_t^{-\frac32} d\inn{P}_t.
\end{align*}
Therefore, the LVR in this scenario is given by
\begin{equation} \label{eqn:CLMM_LRV}
d{\rm LVR}_t = x_t dP_t - dV_t
= \frac14 \ell(P_t) P_t^{-\frac32} d\inn{P}_t.
\end{equation}
Consequently, the impermanent loss ${\rm IL}_t$ can be expressed as
$$
d{\rm IL}_t = (x_0 - x_t)dP_t + d {\rm LVR}_t.
$$

This analysis reveals that, similar to standard CFMMs, the LVR for a liquidity profile is influenced only by the quadratic variation of the pool price $P$ and the liquidity amount at the current price. This highlights the importance of understanding the price dynamics and liquidity distribution when assessing the risks associated with concentrated liquidity provision.

\subsubsection{Time-Dependent Liquidity Profile} \label{section:time-dependent_LP}
When the liquidity profile $\ell_t(\cdot)$ varies over time, the reserve amounts $(x_t, y_t)$ are no longer static. However, the fundamental relationships between these quantities still hold at each point in time. Specifically, for a given price $P$ and time $t$, the reserve amounts are given by
\begin{equation} \label{eqn:CLMM_dynamic_xy}
x_t(P) = \frac12 \int_P^\infty \ell_t(p)p^{-\frac32} dp, \quad
y_t(P) = \frac12 \int_0^P \ell_t(p)p^{-\frac12} dp.  
\end{equation}
The dynamics of the reserve amounts follow from the dynamics of the liquidity profile:
\begin{align*}
dx_t(P) &= \frac12 \int_P^\infty d\ell_t(p)p^{-\frac32} dp, \\
dy_t(P) &= \frac12 \int_0^P d\ell_t(p)p^{-\frac12} dp.
\end{align*}
These equations link the evolution of the liquidity profile to changes in the reserve amounts, capturing the dynamic interplay between liquidity provision and asset holdings.

The instantaneous price $P_t$ continues to be defined as the negative slope of the tangent to the reserve curve at time $t$:
$$
m_t(P) := \frac{\partial y_t}{\partial x_t} = -P.
$$
The dynamics of the instantaneous price are given by
$$
dm_t(P) = d\left(\frac{\partial y_t}{\partial x_t}\right) = d\left(-P\right) = 0.
$$
This equation highlights a crucial property: even with a time-varying liquidity profile, the instantaneous price at a given price level $P$ remains independent of time.

\begin{remark} 
The evolution of $\ell_t$ can be modeled using various stochastic processes, such as Gaussian processes or stochastic partial differential equations (SPDEs), to capture the dynamic nature of liquidity provision. For instance, a stochastic heat equation of the form
$$
d\ell_t = \mathcal{A}_t \ell_t dt + \sigma \ell_t dW_t
$$
could be employed, where $\mathcal{A}_t$ is a second-order parabolic or elliptic differential operator, and $W_t$ is standard Brownian motion. This allows for a flexible and nuanced representation of how liquidity profiles change over time.
\end{remark}

\begin{remark}
While the framework in Section \ref{section:LVR_CLMM} can accommodate time-dependent liquidity profiles, we focus on the time-independent case for the following reasons:
\begin{itemize}
    \item Active liquidity management by LPs introduces non-self-financing aspects, as LPs may need to add or remove assets from the pool to maintain their desired liquidity positions. To address this within the context of IL and LVR analysis, the terms $\frac12 P_t \int_{P_t}^\infty d\ell(p)p^{-\frac32} dp + \frac12 \int_0^{P_t} d\ell_t(p)p^{-\frac12} dp$, which arise from liquidity adjustments, need to be incorporated into the definition of the hedging/replication strategy.
    \item Price changes $dP_t$ and liquidity adjustments $d\ell_t$ occur as distinct, ordered events on the blockchain. This implies that $d\inn{P, \ell(P)}_t = 0$ simplifies the expression for $dV_t$ by eliminating several terms.
    \item Consequently, the core component of $d {\rm LVR}_t$ that is both unhedgeable and directly attributable to arbitrage remains $\frac14 \ell(P_t) P_t^{-\frac32} d\inn{P}_t$, consistent with the static liquidity profile case in Equation \eqref{eqn:CLMM_LRV}. This emphasizes that the primary driver of LVR in CLMMs is the quadratic variation of the price, even in the presence of time-dependent liquidity profiles.
\end{itemize}
\end{remark}

%% file: Arbitrage.tex
\section{Price Process Models and Arbitrage} \label{section:Arbitrage}

\subsection{Continuous Trading Framework} \label{section:conti_trading}
To delve deeper into the dynamics of CLMMs under continuous trading, we construct a mathematical framework that seamlessly integrates order flows, price impacts, fees, and liquidity constraints. This framework serves as the cornerstone for comprehending arbitrage opportunities and constructing optimal trading strategies within the CLMM environment.

\subsubsection{Order Flow Dynamics}
Given a filtered probability space $(\Omega, \mathcal{F}, \{\mathcal{F}_t\}, \mathbb{P})$ that adheres to the usual conditions, the trading dynamics are characterized by two non-negative, adapted processes:
\begin{itemize}
    \item $u^a_t$: The rate at which the CLMM buys asset $X$, capturing the influx of buy orders for the asset.
    \item $u^b_t$: The rate at which the CLMM sells asset $X$, representing the flow of sell orders for the asset.
\end{itemize}
These processes encapsulate the continuous trading activity within the CLMM to model the evolution of the pool's state.

The evolution of the pool state $(x_t, y_t, P_t)$ follows directly from differentiating Equations \eqref{eqn:CLMM_dynamic_xy}:
\begin{align} \label{eqn:CLMM_orderflow_xy}
dx_t &= \frac12 \int_{P_t}^\infty d\ell_t(p) p^{-\frac32} dp + (u^a_t - u^b_t) dt, \notag\\
dy_t &= \frac12 \int_0^{P_t} d\ell_t(p) p^{-\frac12} dp - P_t (u^a_t - u^b_t) dt, \\
dP_t &= \frac{-2 P_t^{\frac{3}{2}}}{\ell_t(P_t)} (u^a_t - u^b_t) dt. \notag
\end{align}
These equations depict how the pool's reserves and price evolve in response to trading activity, capturing the interplay between order flows, liquidity profiles, and price adjustments.

\subsubsection{Fee Considerations}
For a fee tier $(1-\gamma)$, fees accumulate according to
\begin{equation} \label{eqn:CLMM_dynamic_fee}
    dF^x_t = \frac{1-\gamma}{\gamma} u^a_t dt, \quad
    dF^y_t = \frac{1-\gamma}{\gamma} P_t u^b_t dt
\end{equation}
where $F^x_t$ (resp. $F^y_t$) represents the cumulative amount of fees accrued in asset $X$ (resp. $Y$).

\subsection{Arbitrage Models}
This section delves into arbitrage models within the established CLMM continuous trading framework. To ensure analytical tractability and accommodate trading fees \eqref{eqn:CLMM_dynamic_fee}, we assume that the price process, $P_t$, has finite variation. This assumption excludes stochastic differential equations (SDEs) with a diffusion term for the price process.

\subsubsection{Mispricing Process}
We begin by introducing an exogenous price process, $S_t$, for asset $X$, which follows geometric Brownian motion (GBM):
$$
d \ln S_t = \mu dt + \sigma dW_t.
$$
where $\mu$ represents the drift and $\sigma$ the volatility parameter. This exogenous price reflects the asset's true underlying value, which may deviate from the CLMM price due to market inefficiencies or temporary imbalances.

To quantify this discrepancy, we define the mispricing process, $Z_t$, as the logarithmic difference between the CLMM price, $P_t$, and the fair price, $S_t$:
$$
Z_t = \ln S_t - \ln P_t.
$$
Under the finite variation assumption for $P_t$, $Z_t$ must have a diffusion coefficient of  $\sigma$. Consequently, it follows the SDE:
\begin{equation} \label{eqn:SDE_mispricing}
dZ_t = (\mu - u_t) dt + \sigma dW_t,
\end{equation}
where $u_t$ controls the log-price movement rate:
$$
d \ln P_t = \frac{-2 P_t^{\frac{1}{2}}}{\ell_t(P_t)} (u^a_t - u^b_t) dt := u_t dt.
$$

\begin{remark} \
\begin{enumerate}
    \item The exogenous price, $S_t$, typically refers to the price observed on centralized exchanges with high trading volumes, providing a reliable benchmark for the asset's fair value.
    \item The logarithmic scale, $\ln P_t$, for price and liquidity provision aligns with the design of Uniswap V3.
    \item This modeling approach draws an analogy with LOBs \cite{Cont2010StochasticModel, Cont2013PriceDynamics, Cont2021StochasticPartial, ContDegondXuan2023}, where $S_t$ represents the mid-price and the bid-ask spread is determined by the fee parameter $\gamma$. However, unlike traditional LOBs, the bid-ask spread in a CLMM is fixed, and liquidity provision does not directly influence the price.
\end{enumerate}
\end{remark}

\subsubsection{Myopic Arbitrage}
This model explores a scenario where arbitrageurs exclusively drive the CLMM price. We assume a group of arbitrageurs continuously monitors the CLMM and the fair market prices, executing trades whenever profitable arbitrage opportunities arise, specifically when $S_t - \gamma^{-1}P_t > 0$ (buy opportunity) or $\gamma P_t - S_t > 0$ (sell opportunity).  

We introduce the following additional assumptions:
\begin{assumption} \label{Asm:myoptic_arbitrage} \
\begin{itemize}
    \item The reference market exhibits infinite liquidity, implying arbitrage trading does not affect the fair price, $S_t$.
    \item Arbitrageurs are \textbf{myopic}, meaning they execute trades immediately to capitalize on any observed price discrepancies without considering future price movements or potential competition from other arbitrageurs.
\end{itemize}
\end{assumption}
These assumptions lead to the price dynamics described in \cite{fukasawa2023modelfree, najnudel2024arbitrage, lee2024growth}:

\begin{proposition}[Myopic Arbitrage Dynamics] \label{Prop:Myopic_dynamics}
Given a continuous market price $S_t$ satisfying the initial condition $\gamma P_0 \le S_0 \le \gamma^{-1} P_0$, and under Assumptions \ref{Asm:myoptic_arbitrage}, the following hold:
\begin{enumerate}
\item[a)] The mispricing process $Z_t$ can be decomposed as $Z_t = \ln S_t - \ln P_0 + L_t - U_t$ and remains within the range $[\ln \gamma, -\ln \gamma]$ for all $t \geq 0$.
\item[b)] $L_t$ and $U_t$ are both non-decreasing and continuous, with initial values $L_0 = U_0 = 0$.
\item[c)] $L_t$  increases only when $Z_t = - \ln \gamma$, and $U_t$ increases only when $Z_t = \ln \gamma$.
\end{enumerate}
Moreover, $L_t$ and $U_t$ satisfy
\begin{align*}
L_t &= \sup_{0 \le s \le t} \left( - \ln(\gamma P_0) + \ln S_s - U_s \right)^-, \\
U_t &= \sup_{0 \le s \le t} \left( \ln(\gamma^{-1} P_0) - \ln S_s - L_s \right)^-.
\end{align*}
\end{proposition}

The following corollaries follow directly from Proposition \ref{Prop:Myopic_dynamics} and Section \ref{section:conti_trading}:

\begin{corollary}[Inventory Dynamics in Myopic Arbitrage]
Under the same assumptions as in Proposition \ref{Prop:Myopic_dynamics}, the following hold:
\begin{enumerate}
    \item[(a)] $x_t$ and $y_t$ are predictable processes.
    \item[(b)] $x_t$ increases only when $Z_t = \ln \gamma$ and decreases only when $Z_t = - \ln \gamma$. Similarly, $y_t$ increases only when $Z_t = - \ln \gamma$ and decreases only when $Z_t = \ln \gamma$.
    \item[(c)] $x_t$ and $y_t$ are continuous and of bounded variation on any bounded interval in $[0, \infty)$.
    \item[(d)] The arbitrage inventory process can be characterized by
    $$
    d x_t
    = \frac{\ell(P_t) P_t^{-\frac12}}{2} \{ dL_t -  dU_t \}, \quad
    d y_t
    = \frac{\ell(P_t) P_t^{-\frac12}}{2} \{ dU_t -  dL_t \}.
    $$
\end{enumerate}
\end{corollary}

\begin{corollary}[Trading Fee Dynamics in Myopic Arbitrage]
Under the same assumptions as in Proposition \ref{Prop:Myopic_dynamics}, the following hold:
\begin{enumerate}
    \item[(a)] $F^x_t$ and $F^y_t$ are predictable processes.
    \item[(b)] $F^x_t$ increases only when $Z_t = \ln \gamma$ and $F^y_t$ increases only when $Z_t = - \ln \gamma$.
    \item[(c)] $F^x_t$ and $F^y_t$ are continuous and of bounded variation on any bounded interval in $[0, \infty)$.
    \item[(d)] The trading fee process can be characterized by
    $$
    d F^x_t
    = \frac{1-\gamma}{\gamma} \frac{\ell(P_t) P_t^{-\frac12}}{2} dL_t, \quad
    d F^y_t
    = \frac{1-\gamma}{\gamma} \frac{\ell(P_t) P_t^{\frac12}}{2} dU_t.
    $$
\end{enumerate}
\end{corollary}

\begin{remark}
In the absence of fees ($\gamma=1$), the arbitrageur's profit equals the LVR given by Equation \eqref{eqn:CLMM_LRV}, driven by the quadratic variation of the price process. However, with fees, the quadratic variation of $P_t$ vanishes due to the finite variation assumption. In this case, the arbitrageur's profit (or LP's loss) is
\begin{align} \label{eqn:arbitrage_fees}
    d \text{ARB}_t 
    &= (S_t - \gamma^{-1}P_t) dx^+_t + (\gamma P_t - S_t) dx^-_t \\
    &= \frac{\ell(P_t) P_t^{\frac12}}{2} \left[(e^{Z_t} - \gamma^{-1}) dL_t + (\gamma - e^{Z_t}) dU_t \right], \notag
\end{align}
where $dx^+_t = \max\{dx_t, 0\}$ and $dx^-_t = \max\{-dx_t, 0\}$. 

Proposition \ref{Prop:Myopic_dynamics}(c) implies that Equation \eqref{eqn:arbitrage_fees} equals zero, suggesting the myopic arbitrage model may not fully capture the continuous dynamics of AMM prices with fees. This motivates the exploration of more sophisticated models incorporating arbitrageur control mechanisms.
\end{remark}

\subsubsection{Finite-Horizon Arbitrage} \label{section:finite_arbitrage}
Instead of myopic arbitrage, we now consider a single arbitrageur who aims to maximize their profit over a finite time horizon $[0, T]$. This arbitrageur strategically controls the CLMM price movement through continuous trading, assuming no transaction costs.

Specifically, the arbitrageur controls the log-price movement rate, which is modeled as $d \ln P_t = u_t dt$. To formalize this, we define the following classes of admissible controls:
\begin{enumerate}
    \item $\mathcal{A}[t,T] := \{ u:[t,T] \times \Omega \to \mathbb{R} \mid u_s \text{ is $\mathcal{F}_t$-progressively measurable},\ \mathbb{E} [ \int^T_t |u_s|^2 ds ] < +\infty \}$ for any $T>0$;
    \item $\mathcal{A} := \bigcap_{T>0} \mathcal{A}[0,T]$.
\end{enumerate}
With these definitions, the arbitrageur's optimization problem can be formulated as
\begin{equation} \label{eqn:finite_arbitrage}
   \sup_{u \in \mathcal{A}[0,T]} \mathbb{E} \left[ \int_0^T (S_t - P_t) dx_t \right]. 
\end{equation}
Using the relation $dx_t = \frac12 \ell_t(P_t)P_t^{-\frac32} dP_t$ and the approximation $a \approx e^a - 1$ for small $a$, we can approximate the optimization problem \eqref{eqn:finite_arbitrage} in log-scale as
$$
\sup_{u \in \mathcal{A}[0,T]} \mathbb{E} \left[\int_0^T \ell_t(P_t) P_t^{-\frac12} Z_t u_t dt \right],
$$
This approximation simplifies the optimization problem while preserving its essential features, allowing for a more feasible analysis. Recall that the mispricing process, $Z_t$, is defined by Equation \eqref{eqn:SDE_mispricing}.

To enhance the tractability of the control problem, we introduce the following simplifying assumptions:
\begin{assumption} \label{Asm:control_arbitrage} \
\begin{itemize}
    \item (Square Rule for Liquidity) The time-dependent liquidity profile, $\ell_t(P_t)$, adheres to a "square rule" $\ell_t(P_t) P_t^{\frac12} \equiv K$, where $K$ is a constant. This assumption ensures a simple relationship between liquidity and price, simplifying the analysis.
    \item (Quadratic Control Penalty) We impose a quadratic penalty, $\frac{\lambda}{2} u^2_t$, on the control variable, $u_t$, where $\lambda>0$ is a penalty parameter. This penalty discourages excessive control actions by the arbitrageur and promotes smoother, more gradual adjustments to the price.
    \item (Mispricing Penalty) We incorporate a mispricing penalty, $\frac{\tau}{2} Z^2_t$, where $\tau>0$ is another penalty parameter. This penalty penalizes large deviations of the mispricing process, $Z_t$, from zero, ensuring the validity of the approximation $z \approx e^z - 1$ used earlier and bounding the reward function to ensure the stability of the solution.
\end{itemize} 
\end{assumption}
Under these assumptions, the finite-horizon arbitrage problem is simplified to
$$
\sup_{u \in \mathcal{A}[0,T]} \mathbb{E}_z \left[\int_0^T \left\{ Z_t  u_t - \frac{\lambda}{2} u^2_t - \frac{\tau}{2} Z^2_t\right\} dt \right].
$$

\begin{remark} \
\begin{itemize}
    \item The "square rule" can be interpreted as maintaining constant instantaneous liquidity in terms of the numéraire. Recall that in a CPMM, the LP's wealth is represented by $V(P) = 2 \ell P^{\frac12}$ (see Equation \eqref{eqn:G3M_value}). Thus, this condition ensures that the "virtual" wealth associated with the liquidity at the current price remains stable, even as the risky asset's price fluctuates.
    \item From an economic perspective, the mispricing penalty can be viewed as a cost associated with the risk that other arbitrageurs might exploit the price discrepancy before the current arbitrageur can react. This penalty reflects the competitive nature of the arbitrage environment and the potential for missed opportunities due to delayed actions.
    \item While our analysis focuses on the case without transaction costs and trading fees for clarity, a similar analysis can be conducted for the more general case. However, incorporating these factors involves significantly more technical details and is left for future research.
\end{itemize}
\end{remark}

Consider the value function
\begin{equation} \label{eqn:finite_log_arbitrage}
    V(t,z) = V_T(t,z) := \sup_{u \in \mathcal{A}[t,T]} \mathbb{E}_z \left[\int_t^T \left\{ Z_s  u_s - \lambda u^2_s - \frac{\tau}{2} Z^2_s\right\} ds \right].
\end{equation}
The dynamic programming principle leads to the Hamilton-Jacobi-Bellman (HJB) equation:
\begin{equation}
0 = \partial_t V + \sup_{u} \left\{ \frac{\sigma^2}{2} \partial_{zz} V + (\mu - u) \partial_{z} V + z u - \frac{\lambda}{2} u^2 - \frac{\tau}{2} z^2 \right\}
\end{equation}
with the boundary condition $V(T,z)=0$. The optimal control, $u^*$, is
\begin{equation} \label{eqn:optimal_control}
     u^* = \mathop{\arg\max}_u \{- u \partial_z V  + z u - \frac{\lambda}{2} u^2\} = \frac{z - \partial_z V}{\lambda}.
\end{equation}
Substituting this into the HJB Equation yields
\begin{equation} \label{eqn:finite_HJB}
0 = \partial_t V + \frac{\sigma^2}{2} \partial_{zz} V + \mu \partial_z V + \frac{(z - \partial_z V)^2}{2\lambda} - \frac{\tau}{2} z^2
\end{equation}
with terminal condition $V(T,z) = 0$.

\begin{proposition} \label{Prop:finite_verification}
Let $\varphi = \sqrt{\frac{\tau}{\lambda}}$ and $\xi = \frac{1+\sqrt{\lambda\tau}}{1-\sqrt{\lambda\tau}}$. Define
\begin{align*}
h_2(t) =& \sqrt{\lambda \tau} \frac{1 + \xi e^{2\varphi(T-t)}}{1 - \xi e^{2\varphi(T-t)}}, \\
h_1(t) =& \lambda\mu (1+\frac{1}{\sqrt{\lambda\tau}}) \left\{ \frac{e^{\varphi(T-t)} + e^{-\varphi(T-t)} - 2}{\xi e^{\varphi(T-t)} - e^{-\varphi(T-t)}} \right\}, \\
h_0(t) =& \int_t^T \left( \frac{\sigma^2}{2} h_2(s) + \mu h_1(s) + \frac{1}{2\lambda} h^2_1(s) \right) ds. 
\end{align*}
Then, the value function for the control problem \eqref{eqn:finite_log_arbitrage} is
$$
V_T(t,z) = \frac12 h_2(t) z^2 + h_1(t) z + h_0(t).
$$
\end{proposition}

\begin{proof}
Substituting the quadratic ansatz $V(t,z) = \frac12 h_2(t) z^2 + h_1(t) z + h_0(t)$ into the HJB equation \eqref{eqn:finite_HJB} and comparing coefficients, we obtain the following system of ODEs:
\begin{align*}
z^2 &: \dot h_2 + \frac1{\lambda}(1-h_2)^2 - \tau = 0,\\
z &: \dot h_1 + \mu h_2 - \frac1{\lambda} (1 - h_2)h_1 = 0, \\
1 &: \dot h_0 + \frac{\sigma^2}2 h_2 + \mu h_1 + \frac{1}{2\lambda} h_1^2 = 0
\end{align*}
with the terminal conditions $h_2(T) = h_1(T) = h_0(T) = 0$. Solving this system of ODEs yields the expressions for $h_2(t)$, $h_1(t)$, and $h_0(t)$ given in the proposition statement. A standard verification argument, such as the one outlined in \cite[Theorem 3.5.2]{pham2009continuoustime} is indeed the value function for the control problem \eqref{eqn:finite_log_arbitrage}.
\end{proof}

\begin{remark} \
\begin{enumerate}
    \item While the closed-form solution for $h_0(t)$ is generally complex, it can be computed using symbolic integration software like MATLAB. In the specific case where $\mu=0$, the expressions for $h_1(t)$ and $h_0(t)$ simplify significantly:
    $$
    h_1(t)=0, \quad h_2(t) = -\frac{\sigma^2}{2 \varphi} \ln \frac{\xi-1}{\xi e^{\varphi(T-t)} - e^{-\varphi(T-t)}}.
    $$
    \item As $T-t \to 0^+$ and $\lambda \to 0^+$, the optimal control approaches a "bang-bang" type
    $$
    u^* = \left\{\begin{array}{cl}
    +\infty & \mbox{ if } z > 0; \\
    0 & \mbox{ if } z = 0; \\
    -\infty & \mbox{ if } z < 0,
    \end{array}\right.
    $$
    This behavior illustrates why a myopic or impatient trader, whose trading rate is not penalized ($\lambda=0$), would optimally exploit any observed price discrepancy (mispricing) $z$ immediately.
\end{enumerate}
\end{remark}

\subsubsection{Discounted Infinite-Horizon Arbitrage Model} \label{section:discounted_arbitrage}
We now consider a single arbitrageur who aims to maximize their discounted long-term profit:
\begin{equation} \label{eqn:discounted_arbitrage}
\sup_{u \in \mathcal{A}} \mathbb{E} \left[\int_0^\infty e^{-\rho t} (S_t - P_t) dx_t  \right],
\end{equation}
where $\rho$ is the discount rate. Applying similar arguments and assumptions as in Section \ref{section:finite_arbitrage}, we can express the value function in logarithmic scale as
\begin{equation} \label{eqn:discounted_log_arbitrage}
V(z) = V_{\rho}(z) := \sup_{u \in \mathcal{A}} \mathbb{E}_z \left[\int_0^\infty e^{-\rho t} \left\{ Z_t  u_t - \frac{\lambda}{2} u^2_t - \frac{\tau}{2} Z^2_t \right\} dt \right].
\end{equation}
This leads to the following HJB equation
\begin{equation} \label{eqn:discounted_HJB}
    \rho V = \frac{\sigma^2}{2} \partial_{zz} V + \mu \partial_z V + \frac{(z-\partial_z V)^2}{2\lambda} - \frac{\tau}{2} z^2.
\end{equation}

\begin{proposition} \label{Prop:discounted_verification}
Let
\begin{align*}
h_2 &= 1 + \frac{\rho \lambda}{2} - \sqrt{\frac{\rho^2\lambda^2}{4} + \rho \lambda + \tau \lambda}, \\
h_1 &= -\frac{\mu}{\rho+\tau} \left[ (1 + \rho \lambda) h_2 - (1 - \tau \lambda) \right], \\
h_0 &= \frac{1}{\rho} \left\{ \left[ \frac{\sigma^2}{2} + \frac{\rho \mu^2}{(\rho+\tau)^2} \frac{(1 + \rho \lambda)^2}{2} \right] h_2 + \frac{\mu^2}{(\rho+\tau)^2} (1 - \lambda \tau) (\frac{\tau}{2} - \frac{\rho^2\lambda}{2}) \right\}. 
\end{align*}
Then, the value function for the control problem \eqref{eqn:discounted_log_arbitrage} is
$$
V_{\rho}(z) = \frac12 h_2 z^2 + h_1 z + h_0.
$$
\end{proposition}

\begin{proof}
Substituting the quadratic ansatz $V(z) = \frac12 h_2 z^2 + h_1 z + h_0$ into the HJB equation \eqref{eqn:discounted_HJB} and comparing coefficients, we obtain the following system of equations:
\begin{align*}
z^2 &: \ \rho h_2 = \frac{1}{\lambda} (1-h_2)^2 - \tau, \\
z &: \ \rho h_1 = \mu h_2 - \frac{1}{\lambda} (1-h_2)h_1, \\
1 &: \ \rho h_0 = \frac{\sigma^2}{2} h_2 + \mu h_1 + \frac{1}{2\lambda} h_1^2.
\end{align*}
The first equation is a quadratic equation for $h_2$, which has two solutions
$$
h^\pm_2 = 1 + \frac{\rho\lambda}{2} \pm \sqrt{ \frac{\rho^2\lambda^2}{4} + \rho \lambda + \tau \lambda}.
$$
For each solution $h_2^\pm$, we can solve for $h_1^\pm$ and $h_0^\pm$ using the second and third equations, respectively. This gives us two candidate value functions:
$$
V^{\pm}(z) = h_2^\pm z^2 + h_1^\pm z + h_0.
$$
The corresponding optimal controls are
\begin{align*}
u^{\pm}(z) 
&= \frac{1-h_2}{\lambda} z - \frac{\mu}{\rho+\tau} \left(\rho + \tau + \frac{h_2-1}{\lambda} \right) \notag \\
&= \left(\mp \sqrt{\frac{\rho^2}{4} + \frac{\rho}{\lambda} + \frac{\tau}{\lambda}} - \frac{\rho}{2} \right) \left( z - \frac{\mu}{\rho+\tau} \right) + \mu. 
\end{align*}
Substituting these controls into the SDE for the mispricing process \eqref{eqn:SDE_mispricing}, we get
$$
dZ^{\pm}_t = \left(\mu - u^\pm(Z_t)\right) dt + \sigma dW_t 
= \left(\mp \sqrt{\frac{\rho^2}{4} + \frac{\rho}{\lambda} + \frac{\tau}{\lambda}} - \frac{\rho}{2} \right) \left( \frac{\mu}{\rho+\tau} - Z_t \right) dt + \sigma dW_t.
$$

To determine the true value function, we apply the verification theorem \cite[Theorem 3.5.3]{pham2009continuoustime}, which requires that $\mathop{\lim\sup}_{T \to \infty} e^{-\rho T} \mathbb{E}[ V(Z_T)] = 0$. A direct computation reveals that this condition holds for $V^-(z)$ but not for $V^+(z)$. Therefore, the true value function is $V(z) = V^-(z)$, proving the proposition.
\end{proof}

Under the optimal arbitrage strategy
\begin{align} \label{eqn:discounted_optimal_control_explicit}
u^*(z) = \left(\sqrt{\frac{\rho^2}{4} + \frac{\rho}{\lambda} + \frac{\tau}{\lambda}} - \frac{\rho}{2} \right) \left( z - \frac{\mu}{\rho+\tau} \right) + \mu,
\end{align}
the mispricing process $Z_t$ follows an Ornstein-Uhlenbeck (OU) process
$$
dZ_t = \left(\mu - u^*(Z_t)\right) dt + \sigma dW_t 
= \left(\sqrt{\frac{\rho^2}{4} + \frac{\rho}{\lambda} + \frac{\tau}{\lambda}} - \frac{\rho}{2} \right) \left( \frac{\mu}{\rho+\tau} - Z_t \right) dt + \sigma dW_t.
$$

\subsubsection{Ergodic Arbitrage Model}
Finally, we consider an arbitrageur aiming to maximize their long-term average profit, formulated as
\begin{equation} \label{eqn:ergodic_arbitrage}
\sup_{\tilde{u}} \lim_{T \to \infty} \frac{1}{T} \mathbb{E} \left[\int_0^T (S_t - P_t) dx_t  \right].
\end{equation}
Under the previous assumptions and simplifications, this reduces to
\begin{equation} \label{eqn:ergodic_log_arbitrage}
    J(z;u) = \lim_{T \to \infty} \frac{1}{T} \mathbb{E} \left[ \int_0^T  \left\{ Z_t  u_t - \frac{\lambda}{2} u^2_t - \frac{\tau}{2} Z^2_t \right\} dt \right], \quad \eta = \sup_{u \in \mathcal{A}} J(z;u).
\end{equation}

Analyzing this problem by taking limits in the discounted infinite-horizon ($\rho \to 0$) and finite-horizon cases ($T \to \infty$), we obtain:

To analyze the ergodic arbitrage problem \eqref{eqn:ergodic_log_arbitrage}, we follow the general method used in works such as \cite{arisawa1998ergodic, arapostathis2011ergodic}. By taking limits in the discounted infinite-horizon ($\rho \to 0$) and finite-horizon cases ($T \to \infty$), we obtain:

\begin{theorem} \label{Thm:ergodic_constant}
There exists a constant $\hat{\eta} \in \mathbb{R}$ such that for all $z \in  \mathbb{R}$,
$$
\hat{\eta} = \lim_{\rho \to 0} \rho V_{\rho}(z) = \lim_{T \to \infty} \frac{1}{T} V_T(0,z)
$$
where $V_{\rho}$ and $V_T$ are the value functions defined in \eqref{eqn:discounted_log_arbitrage} and \eqref{eqn:finite_log_arbitrage}, respectively. Moreover,
$$
\eta = \hat{\eta} = \frac{\sigma^2}{2} + \frac{\mu^2}{2} \left(\frac{1}{\tau} - \lambda \right).
$$
\end{theorem}

\begin{proof}
We first show that $\lim_{\rho \to 0} \rho V_{\rho}(z)$ exists. From the explicit expression for $V_{\rho}$ derived in Proposition \ref{Prop:discounted_verification}, we have
\begin{align*}
\lim_{\rho \to 0} \rho V_{\rho}(z)
= \lim_{\rho \to 0} \rho \left[ \frac{1}{2} h_2 z^2 + h_1 z + h_0 \right]
= \frac{\sigma^2}{2} + \frac{\mu^2}{2} \left( \frac{1}{\tau} - \lambda \right).
\end{align*}
This limit exists and is the same for all $z \in \mathbb{R}$.

Next, we establish that $\lim_{T \to \infty} \frac{1}{T} V_T(0,z)$ exists. From the explicit expression for $V_T$ from Proposition \ref{Prop:finite_verification}\footnote{To verify the case when $\mu \neq 0$, we utilize the explicit closed-form solution computed by MATLAB.}, we have
\begin{align*}
\lim_{T \to \infty} \frac{1}{T} V_T(0, z)
= \lim_{T \to \infty} \frac{1}{T} \left[ \frac{1}{2} h_2(0) z^2 + h_1(0) z + h_0(0) \right]
= \frac{\sigma^2}{2} + \frac{\mu^2}{2} \left( \frac{1}{\tau} - \lambda \right).
\end{align*}
Again, this limit also exists and is finite for all $z \in \mathbb{R}$.

Finally, we need to show that $\eta = \hat{\eta}$, which is a consequence of \cite[Lemma 3.6.4]{arapostathis2011ergodic}.
\end{proof}

\begin{corollary}
The optimal Markov control for the ergodic control problem \eqref{eqn:discounted_log_arbitrage} is
$$
u^*(z) = \sqrt{\frac{\tau}{\lambda}} z + \left(1-\frac{1}{\sqrt{\lambda\tau}}\right) \mu.
$$
Furthermore, under this control, the mispricing process is driven by the OU process
$$
dZ_t = \sqrt{\frac{\tau}{\lambda}} \left(\frac{\mu}{\tau} - Z_t \right) dt + \sigma dW_t.
$$
\end{corollary}

\begin{proof}
According to \cite[Theorem 3.6.6]{arapostathis2011ergodic}, the ergodic HJB equation for the control problem \eqref{eqn:discounted_log_arbitrage} is
$$
\eta = \sup_{u} \left\{ \frac{\sigma^2}{2} \partial_{zz} V + (\mu - u) \partial_z V + z u - \frac{\lambda}{2} u^2 - \frac{\tau}{2} z^2 \right\}
$$
where $\eta$ is the ergodic constant given in Theorem \ref{Thm:ergodic_constant} and $V$ is the relative value function, uniquely defined up to a constant. The supremum is attained at $u^*$, given again by equation \eqref{eqn:optimal_control}.

Using the ansatz $V(z) = \frac{1}{2} h_2 z^2 + h_1 z + h_0$, we obtain the following system of equations:
\begin{align*}
\frac{1}{\lambda} (1-h_2)^2 - \tau &= 0, \\
\mu h_2 - \frac{1}{\lambda} (1-h_2)h_1 &= 0, \\
\frac{\sigma^2}{2} h_2 + \mu h_1 + \frac{1}{2\lambda} h_1^2 &= \eta.
\end{align*}
To ensure that the optimal control belongs to the set of admissible controls, we select the root $h_2 = 1 -\sqrt{\lambda\tau}$.  This leads to $h_1 = \frac{\lambda \mu h_2}{1 - h_2} = \mu (\sqrt{\frac{\lambda}{\tau}} - \lambda)$. Therefore, the optimal ergodic control in feedback form is
$$
u^*(z) = \frac{z-\partial_z V}{\lambda} = \frac{(1-h_2)z - h_1}{\lambda} = \sqrt{\frac{\tau}{\lambda}} z + \left(1-\frac{1}{\sqrt{\lambda\tau}}\right) \mu.
$$
The SDE for the mispricing process, $Z_t$, follows immediately from substituting this control into equation \eqref{eqn:SDE_mispricing}.
 \end{proof}

%% file: Conclusion.tex
\section{Conclusion} \label{section:Conclusion}
This paper establishes a comprehensive mathematical framework for analyzing Concentrated Liquidity Market Makers (CLMMs) in continuous time, offering valuable insights into these essential DeFi primitives. We introduce a novel approach to modeling liquidity profiles as measure-valued processes, precisely characterizing how concentrated liquidity affects market behavior and trading outcomes. Our investigation of arbitrage dynamics reveals that trading fees fundamentally constrain admissible price processes to those with finite variation, preventing infinite fee generation. This finding has significant implications for CLMM design and highlights the need to adapt traditional continuous-time finance models for DeFi applications.

Furthermore, we derive closed-form solutions for optimal arbitrage strategies under three distinct scenarios: myopic arbitrage, finite-horizon optimization, and infinite-horizon optimization with discounted and ergodic controls. These solutions provide a deeper understanding of how rational actors interact with CLMMs, influencing price discovery and market efficiency. Notably, our analysis demonstrates that optimal arbitrage in the ergodic case leads to an Ornstein-Uhlenbeck process for the mispricing process, suggesting a natural mean-reversion tendency in CLMM markets.

Our framework illuminates practical considerations for CLMM design. The constraint on price processes imposed by trading fees necessitates careful calibration of fee parameters to balance revenue generation with market efficiency. Additionally, the relationship between concentrated liquidity provision and market stability underscores the importance of well-designed incentive mechanisms for optimal liquidity distribution. Our characterization of price process behavior under different arbitrage models offers valuable guidance for developing robust and manipulation-resistant price oracles.

Future research directions include incorporating transaction costs and trading fees into the arbitrage strategy analysis, developing optimal liquidity provision strategies under various market conditions, investigating the impact of discrete block times on continuous-time approximations, analyzing the stability and convergence of complex fee structures, and studying the interactions between diverse market participants. As DeFi continues to evolve, this rigorous approach to analyzing market mechanisms will become increasingly crucial for informed protocol design and optimization.

\section*{Acknowledgement}
We gratefully acknowledge the valuable feedback received from participants at seminars and lectures delivered at Bocconi University, the Dipartimento di Scienze per l'Economia e l'Impresa at the Università degli Studi di Firenze, and Ritsumeikan University. We extend special thanks to Jimmy Risk for their contributions to generating figures from empirical liquidity profiles and to Wun-Cing Liou for his assistance with Uniswap data collection.

S.-N. T. acknowledges the support of the National Science and Technology Council of Taiwan under grant number 111-2115-M-007-014-MY3. He expresses his gratitude to Shuenn-Jyi Sheu and Wei-Cheng Wang for their insightful discussions and assistance, and to Jie-Hong Lai for his help with MATLAB. T.-H. W. expresses his gratitude for the hospitality of the Mathematics Division of the National Center for Theoretical Sciences during his stay in Taipei in August 2024.

%% file: main.bbl
\newcommand{\etalchar}[1]{$^{#1}$}
\begin{thebibliography}{MMRZ22b}

\bibitem[ABG11]{arapostathis2011ergodic}
Ari Arapostathis, Vivek~S. Borkar, and Mrinal~K. Ghosh.
\newblock {\em Ergodic Control of Diffusion Processes}.
\newblock Encyclopedia of Mathematics and its Applications. Cambridge University Press, Cambridge, 2011.

\bibitem[AC20]{Angeris2020ImprovedPO}
Guillermo Angeris and Tarun Chitra.
\newblock Improved price oracles: Constant function market makers.
\newblock In {\em Proceedings of the 2nd ACM Conference on Advances in Financial Technologies}, AFT '20, pages 80--91, New York, NY, USA, 2020. Association for Computing Machinery.

\bibitem[ACD{\etalchar{+}}23]{Angeris2023GeometryCFMM}
Guillermo Angeris, Tarun Chitra, Theo Diamandis, Alex Evans, and Kshitij Kulkarni.
\newblock The geometry of constant function market makers.
\newblock {\em arXiv preprint arXiv:2308.08066}, aug 2023.

\bibitem[ACE22]{angeris2022when}
Guillermo Angeris, Tarun Chitra, and Alex Evans.
\newblock When does the tail wag the dog? curvature and market making.
\newblock {\em Cryptoeconomic Systems}, 2(1), June 2022.

\bibitem[AKC{\etalchar{+}}21]{Angeris2021Analysis}
Guillermo Angeris, Hsien-Tang Kao, Rei Chiang, Charlie Noyes, and Tarun Chitra.
\newblock An {Analysis} of {Uniswap} markets.
\newblock {\em Cryptoeconomic Systems}, 0(1), apr 5 2021.
\newblock https://cryptoeconomicsystems.pubpub.org/pub/angeris-uniswap-analysis.

\bibitem[AL98]{arisawa1998ergodic}
M.~Arisawa and P.-L. Lions.
\newblock On ergodic stochastic control.
\newblock {\em Communications in Partial Differential Equations}, 23(11-12):2187--2217, 1998.

\bibitem[AZR20]{Adams2020UniswapV2}
Hayden Adams, Noah Zinsmeister, and Dan Robinson.
\newblock Uniswap v2 core.
\newblock \url{https://uniswap.org/whitepaper-v2.pdf}, 2020.
\newblock Accessed: \today.

\bibitem[AZS{\etalchar{+}}21]{Adams2021UniswapV3}
Hayden Adams, Noah Zinsmeister, Moody Salem, River Keefer, and Dan Robinson.
\newblock Uniswap v3 core, 2021.
\newblock Accessed: \today.

\bibitem[BBB{\etalchar{+}}24]{bergault2024priceaware}
Philippe Bergault, Louis Bertucci, David Bouba, Olivier Guéant, and Julien Guilbert.
\newblock Price-aware automated market makers: Models beyond brownian prices and static liquidity.
\newblock {\em arXiv e-prints}, May 2024.

\bibitem[CdL13]{Cont2013PriceDynamics}
Rama Cont and Adrien de~Larrard.
\newblock Price dynamics in a {M}arkovian limit order market.
\newblock {\em SIAM Journal on Financial Mathematics}, 4(1):1--25, 2013.

\bibitem[CDSB{\etalchar{+}}23]{cartea2023automated}
Álvaro Cartea, Fayçal Drissi, Leandro Sánchez-Betancourt, David Siska, and Lukasz Szpruch.
\newblock Automated market makers designs beyond constant functions.
\newblock {\em SSRN}, May 2023.

\bibitem[CDX23]{ContDegondXuan2023}
Rama Cont, Pierre Degond, and Lifan Xuan.
\newblock A mathematical framework for modeling order book dynamics.
\newblock {\em arXiv e-prints}, page arXiv:2302.01169, February 2023.

\bibitem[CJ21]{Capponi2021AdoptionDB}
Agostino Capponi and Ruizhe Jia.
\newblock The adoption of blockchain-based decentralized exchanges.
\newblock {\em arXiv preprint arXiv:2103.08842}, 2021.

\bibitem[CM21]{Cont2021StochasticPartial}
Rama Cont and Marvin~S. M\"{u}ller.
\newblock A stochastic partial differential equation model for limit order book dynamics.
\newblock {\em SIAM Journal on Financial Mathematics}, 12(2):744--787, 2021.

\bibitem[CST10]{Cont2010StochasticModel}
Rama Cont, Sasha Stoikov, and Rishi Talreja.
\newblock A stochastic model for order book dynamics.
\newblock {\em Operations Research}, 58(3):549--563, 2010.

\bibitem[EGM23]{echenim2023thorough}
Mnacho Echenim, Emmanuel Gobet, and Anne-Claire Maurice.
\newblock Thorough mathematical modelling and analysis of uniswap v3.
\newblock Working paper or preprint, September 2023.

\bibitem[Eva21]{Evans2021Liquidity}
Alex Evans.
\newblock Liquidity {Provider} {Returns} in {Geometric} {Mean} {Markets}.
\newblock {\em Cryptoeconomic Systems}, 1(2), oct 22 2021.
\newblock https://cryptoeconomicsystems.pubpub.org/pub/evans-g3m-returns.

\bibitem[FMW23a]{fukasawa2023modelfree}
Masaaki Fukasawa, Basile Maire, and Marcus Wunsch.
\newblock Model-free hedging of impermanent loss in geometric mean market makers.
\newblock {\em arXiv e-prints}, 2023.

\bibitem[FMW23b]{fukasawa2023weighted}
Masaaki Fukasawa, Basile Maire, and Marcus Wunsch.
\newblock Weighted variance swaps hedge against impermanent loss.
\newblock {\em Quantitative Finance}, 23(6):901--911, 2023.

\bibitem[GM23]{Gobet2023DecentralizedFB}
Emmanuel Gobet and Anastasia Melachrinos.
\newblock Decentralized finance \& blockchain technology.
\newblock In {\em SIAM Financial Mathematics and Engineering 2023}, Philadelphia, United States, Jun 2023.

\bibitem[GPW{\etalchar{+}}13]{Gould2013LimitOB}
Martin~D. Gould, Mason~A. Porter, Stacy Williams, Mark McDonald, Daniel~J. Fenn, and Sam~D. Howison.
\newblock Limit order books.
\newblock {\em Quantitative Finance}, 13(11):1709--1742, 2013.

\bibitem[HRS{\etalchar{+}}21]{harvey2021defi}
C.R. Harvey, A.~Ramachandran, J.~Santoro, F.~Ehrsam, and V.~Buterin.
\newblock {\em DeFi and the Future of Finance}.
\newblock Wiley, 2021.

\bibitem[HSW23]{heimbach2023risks}
Lioba Heimbach, Eric Schertenleib, and Roger Wattenhofer.
\newblock Risks and returns of uniswap v3 liquidity providers.
\newblock In {\em Proceedings of the 4th ACM Conference on Advances in Financial Technologies}, AFT '22, pages 89--101, New York, NY, USA, 2023. Association for Computing Machinery.

\bibitem[Loe22]{Loesch2022QuantitativeFinance}
Stefan Loesch.
\newblock The quantitative finance aspects of automated market makers in defi.
\newblock {\em arXiv preprint arXiv:2212.10974}, December 2022.

\bibitem[LTW24]{lee2024growth}
Cheuk~Yin Lee, Shen-Ning Tung, and Tai-Ho Wang.
\newblock Growth rate of liquidity provider's wealth in g3ms.
\newblock {\em arXiv e-prints}, 2024.

\bibitem[MM19]{Martinelli2019Balancer}
Fernando Martinelli and Nikolai Mushegian.
\newblock Balancer: A non-custodial portfolio manager, liquidity provider, and price sensor.
\newblock \url{https://balancer.fi/whitepaper.pdf}, 2019.
\newblock Accessed: \today.

\bibitem[MMR23a]{Milionis2023AutomatedMM}
Jason Milionis, Ciamac~C. Moallemi, and Tim Roughgarden.
\newblock Automated market making and arbitrage profits in the presence of fees.
\newblock {\em arXiv preprint arXiv:2305.14604}, 2023.

\bibitem[MMR23b]{Milionis2023complexityapproximation}
Jason Milionis, Ciamac~C. Moallemi, and Tim Roughgarden.
\newblock Complexity-approximation trade-offs in exchange mechanisms: Amms vs. lobs.
\newblock In {\em Financial Cryptography and Data Security: 27th International Conference, FC 2023, Bol, Brač, Croatia, May 1–5, 2023, Revised Selected Papers, Part I}, pages 326--343, Berlin, Heidelberg, 2023. Springer-Verlag.

\bibitem[MMRZ22a]{Milionis2022AutomatedMM}
Jason Milionis, Ciamac~C. Moallemi, Tim Roughgarden, and Anthony~Lee Zhang.
\newblock Automated market making and loss-versus-rebalancing.
\newblock {\em arXiv preprint arXiv:2208.06046}, 2022.

\bibitem[MMRZ22b]{Milionis2022Quantifying}
Jason Milionis, Ciamac~C. Moallemi, Tim Roughgarden, and Anthony~Lee Zhang.
\newblock Quantifying loss in automated market makers.
\newblock In {\em Proceedings of the 2022 ACM CCS Workshop on Decentralized Finance and Security ({DeFi'22})}, pages 71--74, New York, NY, USA, 2022. Association for Computing Machinery.

\bibitem[NTYY24]{najnudel2024arbitrage}
Joseph Najnudel, Shen-Ning Tung, Kazutoshi Yamazaki, and Ju-Yi Yen.
\newblock An arbitrage driven price dynamics of automated market makers in the presence of fees.
\newblock {\em Frontiers of Mathematical Finance}, 3(4):560--571, 2024.

\bibitem[Pha09]{pham2009continuoustime}
Huyên Pham.
\newblock {\em Continuous-time Stochastic Control and Optimization with Financial Applications}.
\newblock Stochastic Modelling and Applied Probability. Springer Berlin, Heidelberg, 2009.

\bibitem[UBY24]{urusov2024backtesting}
Andrey Urusov, Rostislav Berezovskiy, and Yury Yanovich.
\newblock Backtesting framework for concentrated liquidity market makers on uniswap v3 decentralized exchange.
\newblock {\em arXiv e-prints}, October 2024.

\end{thebibliography}
